\documentclass[journal,twoside]{IEEEtran}
\usepackage{amsmath,amsfonts,amssymb}
\usepackage{amsthm}
\usepackage{algorithmic}
\usepackage{algorithm}
\usepackage{array}
\usepackage[caption=false,
font=normalsize,
labelfont=sf,
textfont=sf]{subfig}
\usepackage{textcomp}
\usepackage{stfloats}
\usepackage{url}
\usepackage{verbatim}
\usepackage{graphicx}
\usepackage{cite}
\usepackage{bm}
\usepackage{xcolor}
\usepackage[normalem]{ulem}
\usepackage{multirow}
\usepackage{makecell}
\usepackage{arydshln} 
\usepackage{pgfplots}
\newtheoremstyle{italicstyle}
  {}  
  {}   
  {\itshape}  
  {}         
  {\itshape}  
  {.}        
  {.5em}      
  {}         

\theoremstyle{italicstyle}
\newtheorem{remark}{Remark}
\newtheorem{lemma}{Lemma}
\newtheorem{proposition}{Proposition}
\newtheorem{theorem}{Theorem}

\newtheorem{assumption}{Assumption}

\begin{document}

\title{Physics-based Online Adaptive Koopman Model Predictive Attitude Control for Combined Spacecraft with Dynamic Uncertainties}

\author{Yicheng Sun,
    Yueyong Lyu,
    Yuhan Liu,
    Yanning Guo,
    and Wei Pan
    \thanks{Yicheng Sun and Yanning Guo are with Zhengzhou Advanced Research Institute, Harbin Institute of Technology, Zhengzhou 450000, China, and also with the Department of Control Science and Engineering, Harbin Institute of Technology, Harbin 150001, China (e-mail: sunyc@stu.hit.edu.cn; guoyn@hit.edu.cn).}
    \thanks{Yueyong Lyu and Yuhan Liu are with the Department of Control Science and Engineering, Harbin Institute of Technology, Harbin 150001, China (e-mail: lvyy@hit.edu.cn; yhliu@hit.edu.cn).}
    \thanks{Wei Pan is with the School of Engineering, Newcastle University, Newcastle upon Tyne NE1 7RU, United Kingdom (e-mail: wei.pan2@newcastle.ac.uk).}
    }

\maketitle

\begin{abstract}
This paper proposes a physics-based adaptive Koopman Model Predictive Control (MPC) strategy for combined spacecraft attitude stabilization under inertia uncertainties and active target maneuverability. A novel, quaternion-based Koopman model is constructed from a set of analytical lifting functions derived from the quaternion kinematics, 
which provides a more compact and physically interpretable linear representation of the nonlinear dynamics compared with the conventional black-box EDMD and higher-dimensional DCM-based model. Leveraging the linear structure of this nominal model, a gradient descent-based update law is employed to efficiently identify time-varying inertial uncertainties from real-time input/output data.
By integrating this adaptive linear model into the MPC framework, the optimal control problem reduces to a computationally efficient Quadratic Program (QP), thereby significantly lowering the online computational burden compared to nonlinear adaptive MPC.
Recursive feasibility and regional input-to-state stability are formally established through the design of terminal ingredients for the MPC. The effectiveness and superiority of the proposed strategy are validated through comparative simulations of an attitude stabilization task for combined spacecraft in a high-fidelity 3D simulator.
\end{abstract}

\def\abstractname{Note to Practitioners}
\begin{abstract}
On-orbit servicing missions such as refueling and debris removal face a critical challenge: once a spacecraft docks with another target, its inertial properties change abruptly and unpredictably. Furthermore, if the target is non-cooperative, 
it may also apply competitive torques that actively destabilize the combined system. Such unknown physical variations can endanger the entire mission.
Conventional control schemes typically depend on complex nonlinear models of the spacecraft, which impose a heavy computational burden and fail to achieve rapid on-orbit adjustments.
To overcome these limitations, this work introduces a unified, online adaptive strategy built upon Koopman operator theory and Model Predictive Control (MPC). The proposed strategy delivers a significantly faster and more precise control response under inertial uncertainty and competitive torques, as validated in high-fidelity simulations.
Beyond spacecraft applications, the strategy is potentially applicable to a broader class of nonlinear dynamic systems, especially those with structured kinematics and significant parametric uncertainties.
In particular, the proposed physics-enhanced Koopman-based modeling approach can be directly extended to engineering fields such as unmanned aerial vehicle attitude control and robotic manipulator operations.
\end{abstract}


\section{Introduction}
\IEEEPARstart{O}{n}-orbit servicing (OOS) is rapidly emerging as a cornerstone for sustainable space operations, promising revolutionary capabilities such as satellite life extension, active debris removal, and in-space assembly \cite{rousso2021mission}, \cite{zhihui2021review}, \cite{liu2022neural}. 
A critical phase in many of these missions involves the capture of a target. Following a successful capture, the service and target spacecraft form a single combined body, whose attitude must be precisely controlled to stabilize any tumbling motion to ensure safe operations of follow-up missions.
However, a major challenge arises from the uncertainty in the post-capture dynamics, as the inertial properties—such as the mass, center of mass, and moments of inertia—of a non-cooperative target are typically unknown or poorly estimated.
This significant change in the combined spacecraft's parameters renders traditional model-based attitude controllers inadequate and potentially unstable, creating a critical need for robust and adaptive control strategies that can ensure mission safety and success in such uncertain conditions.
\begin{figure*}[htbp]
	\centering
	\includegraphics[scale=1]{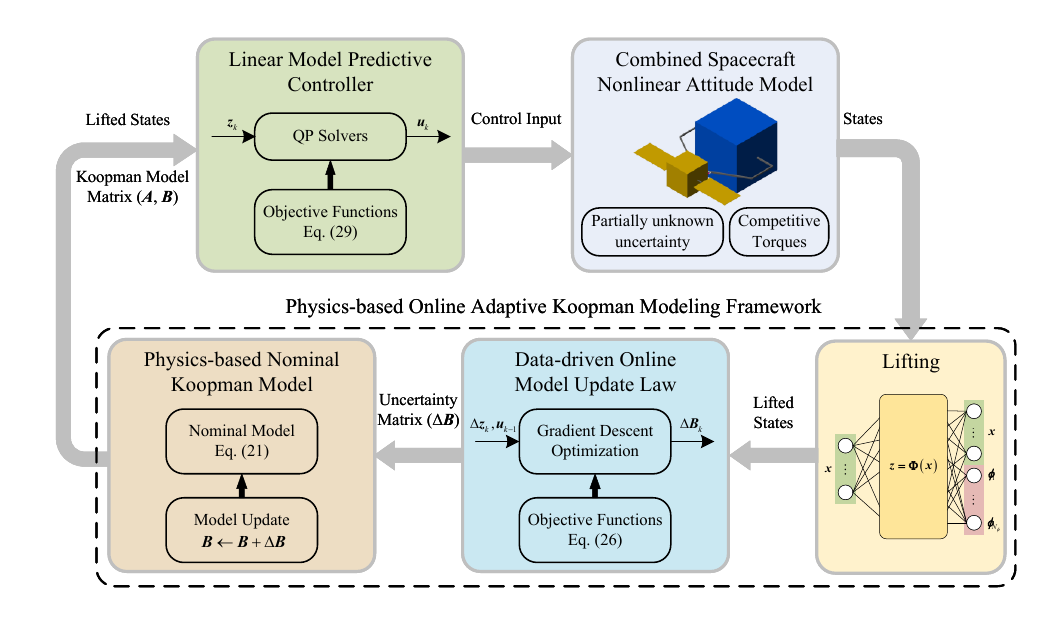}
        \caption{Block diagram of the proposed control framework}
        \label{fig:Blk_framework}
\end{figure*}
In response, a range of advanced model-based strategies have been applied in attitude control and other aerospace applications. 
Jia et al. \cite{jia2025high} applies the HOFA system concept to attitude stabilization for underactuated spacecraft.
In \cite{liu2023observer}, a nonsingular terminal sliding mode controller embedded with a fuzzy logic system is designed for quadrotor Unmanned Aerial Vehicles (UAVs).
\cite{shao2023fault} proposed a barrier Lyapunov functions-based adaptive robust fault-tolerant control (FTC) algorithm for uncertain spacecraft attitude tracking.
Note that a common thread of these strategies is the integration of Disturbance Observer (DO) to compensate for uncertainties, which improves control accuracy but also increases the model dependency of the control strategy and introduces additional computational costs.
\IEEEpubidadjcol
As an advanced model-based control strategy, Model Predictive Control (MPC) \cite{garcia1989model} has emerged as a premier strategy for spacecraft attitude, mainly due to its ability to handle operational constraints, such as actuator saturation and pointing restrictions, while optimizing control performance over a predictive horizon. 
Subsequently, various advanced MPC formulations have been proposed and many of which have been leveraged to enhance attitude control performance under different scenarios.
For instance, stochastic MPC has been employed to manage uncertainties and external disturbances \cite{zhou2025hybrid}, while Tube-based MPC provides a robust framework for handling system constraints \cite{chai2021dual} and enhancing disturbance rejection by being combined with DO \cite{zhang2023robust}. Furthermore, Distributed MPC has emerged as a key enabler for attitude control under multiple optimization objectives \cite{wu2024multi}.
However, the potential of MPC is severely challenged when applied to dynamically complex systems. The most direct approach is Nonlinear MPC (NMPC) \cite{zhou2025hybrid}, \cite{chai2021dual}, \cite{cai2023autonomous}, which uses the full nonlinear model for prediction, but in turn imposes a severe on-board computational burden and makes real-time implementation intractable. Another common approach is to combine Linear MPC (LMPC) with feedback linearization \cite{khodaverdian2023attitude}, \cite{zhao2014quadcopter}, which renders the nonlinear dynamics linear by a nonlinear control input. However, feedback linearization is a fundamentally local technique and relies heavily on a precise physical model for its reliability \cite{gadginmath2024data}.
This creates a dilemma for modern spacecraft attitude control, i.e., how to develop a predictive framework that can capture a system's global nonlinear dynamics without sacrificing real-time computational efficiency.

Originally introduced in 1931 \cite{koopman1931hamiltonian}, the Koopman operator theory offers an alternative perspective for analyzing nonlinear systems. Instead of approximating the system's dynamics, it 'lifts' the state into a higher-dimensional space of observables where the evolution is governed by a linear operator. This powerful transformation enables the application of advanced linear systems theory to analyze nonlinear dynamics.
The linear Koopman operator can be computed directly from the nonlinear dynamics \cite{asada2023global}, \cite{Chen2023}, \cite{zinage2021koopman}, or more commonly, computed from time-series data using data-driven methods like Dynamic Mode Decomposition (DMD) \cite{schmid2010dynamic} and Extended DMD (EDMD) \cite{williams2015data}. 
By approximating nonlinear dynamics with a computationally efficient linear model in a lifted state space, the EDMD-based Koopman framework is naturally suited for MPC \cite{korda2018linear}, \cite{zhang2022robust} and has been successfully applied to complex problems, like robots control \cite{bruder2020data}, and missiles guidance \cite{zhou2024koopman}, and mobile manipulator control \cite{ren2025koopman}.
However, its performance critically depends on the selection of appropriate basis (or lifting) functions. Recent research has moved beyond the selection of traditional polynomials or radial basis functions, shifting towards automated and adaptive methods, such as Deep Neural Networks (DNNs) \cite{xiao2022deep} and Gaussian Processes (GPs) \cite{majumdar2025inverted}, \cite{zhang2026physics}. Furthermore, to address model mismatch or uncertainty, online adaptive Koopman frameworks \cite{singh2025adaptive}, \cite{chen2024learning} have been developed, in which a nominal Koopman model is computed offline and then updated during real-time operation.
Despite this progress, purely data-driven approaches still suffer from a lack of physical interpretability and the curse of dimensionality \cite{schaller2023towards}. Conversely, purely model-based methods, while interpretable, are highly sensitive to model mismatch. This trade-off motivates an adaptive and physics-based data-driven approach. In contrast to generic data-driven liftings, the present work exploits the specific quaternion and rigid-body structure of spacecraft attitude dynamics to construct the lifted model analytically.

In this article, we propose a physics-based adaptive Koopman MPC (denoted by PAKMPC) strategy for combined spacecraft attitude stabilization under dynamic uncertainties, including inertia uncertainties and competitive torque exerted by the target. The block diagram of the overall control framework is shown in Fig. \ref{fig:Blk_framework}.
The main contributions of this article are as follows.
\begin{enumerate}
    \item A novel physics-based Koopman model for combined spacecraft is developed by analytically constructing the lifting functions from the quaternion-based kinematics, rather than selecting them from a generic data-driven dictionary. This yields a compact $(4N_{\phi}+7)$-dimensional representation by incorporating $N_{\phi}$ observable functions, providing superior physical interpretability and reduced structural redundancy over models from data-driven EDMD (up to hundreds of dimensions) \cite{korda2018linear} and traditional DCM-based approaches ($9(N_{\phi}+1)$ dimensions) \cite{Chen2023}.
    \item A new adaptive Koopman-MPC strategy is proposed for combined spacecraft attitude maneuvers under post-capture uncertainty, where inertia mismatch and target-induced disturbances are incorporated as parameter variations within the Koopman model and updated online in a recursive manner. This enables predictive control to be performed within a QP-based formulation, thereby improving computational efficiency compared with the nonlinear adaptive MPC.
    \item Formal guarantees of recursive feasibility and regional input-to-state stability on a compact operating set are established by incorporating a designed terminal cost and constraint set, which provides theoretical support for robust closed-loop operation under online model adaptation.  
    \item The effectiveness and superiority of the proposed strategy are validated in high-fidelity simulations of combined spacecraft attitude stabilization scenarios under inertia uncertainty and competitive target disturbances.
\end{enumerate}

The structure of the article is as follows. Section \ref{Sec_Pre} introduces the attitude dynamics of the combined spacecraft. Section \ref{Sec_Koopman} presents a brief review of the Koopman Operator theory, the construction process of the quaternion-based nominal Koopman model, and the detailed online update law of the Koopman-based model. Section \ref{Sec_Controller} describes the design of online Koopman-based LMPC and proof of recursive feasibility and stability analysis. Section \ref{Sec_Simulation} presents the numerical simulation results. Finally, Section \ref{Sec_Conclusion} concludes the article.

\textbf{Notations}: In this paper, the $n$-dimensional Euclidean space is represented by $\mathbb{R}^n$, and the space of real matrices with $m$ rows and $n$ columns is denoted by $\mathbb{R}^{m \times n}$. $\mathbb{R}_0^{+}$ denotes the set of nonnegative real numbers. $\mathbb{S}^{n \times n}$ is the set of $n \times n$ symmetric matrices. The identity matrix of size $n \times n$ is denoted by $\textbf{I}_n$. $\text{diag}(\mathbf{v})$ denotes a diagonal matrix for a vector $\mathbf{v} = [v_1, v_2, \dots, v_n]^{\top} \in \mathbb{R}^n$. $\text{vec}(\cdot)$ transforms a matrix into a single column vector by stacking its columns vertically. The operator $(\cdot)^{\times}$ constructs a skew-symmetric matrix such that $(\bm{a})^{\times}\bm{b} = \bm{a} \times \bm{b}$ for any vector $\bm{a},\bm{b} \in \mathbb{R}^3$. The set of unit quaternions, which represents rotations in three-dimensional space, is denoted by $\mathbb{Q}^3$. $\text{sgn}(\cdot)$ represents the standard signum function.
\section{Preliminaries} \label{Sec_Pre}
\subsection{Attitude Dynamics of Combined Spacecraft}

The attitude kinematics of the combined spacecraft are described in terms of quaternions by:
\begin{equation} \label{AttKin}
    \dot{\bm{q}} = \bm{\Omega}(\bm{\omega}){\bm{q}}
\end{equation}
where $\bm{q} = \begin{bmatrix} q_0 & \bm{q}_v^{\top} \end{bmatrix}^\top\in \mathbb{Q}^3$ is the unit quaternion representing the orientation of the body frame $\mathcal{F}^{B}$ with respect to the inertial frame $\mathcal{F}^{I}$, $\bm{\omega} =\begin{bmatrix}\omega_1 & \omega_2 & \omega_3 \end{bmatrix}^{\top}\in \mathbb{R}^3$ is the spacecraft angular velocity. The matrix function $\bm\Omega(\bm{\omega}) \in \mathbb{R}^{4 \times 4}$ is defined as
\[
    \bm{\Omega}(\bm{\omega}) = 
    \frac{1}{2}\begin{bmatrix}
    0 & -\omega_1 & -\omega_2 & -\omega_3 \\
     \omega_1 &  0 & \omega_3 & -\omega_2 \\
     \omega_2 & -\omega_3 & 0 & \omega_1  \\
     \omega_3 & \omega_2 & -\omega_1 & 0  \\
    \end{bmatrix}    
\]

The attitude dynamics are governed by the Euler equation for a rigid body:
\begin{equation} \label{AttDyn}
    \bm{J}_{c} \dot{\bm{\omega}} = - \bm{\omega}^{\times}\bm{J}_{c}\bm{\omega} + \bm{\tau} + \bm{\tau}_d
\end{equation}
where $\bm{J}_{c} \in \mathbb{S}^{3 \times 3}$ is the inertia matrix of the combined spacecraft, $\bm{\tau} \in \mathbb{R}^{3}$ is the control torque, $\bm{\tau}_d$ is the external disturbance torque accounting for environmental perturbations and unmodeled effects.

The attitude error $\tilde{\bm{q}} = \begin{bmatrix} \tilde{q}_0 & \tilde{\bm{q}}_v^{\top} \end{bmatrix}^{\top}$ is computed
as $\tilde{{\bm{q}}} = \bigl(\bm{q}^{\text{ref}}\bigr)^{*} \otimes \bm{q}$, where $\bm{q}^{\text{ref}}$ is the desired attitude, $(\cdot)^{*}$ denotes the quaternion conjugate, and "$\otimes$" represents quaternion multiplicatdistion. The angular velocity error is defined as $\tilde{\bm{\omega}} = \bm{\omega} - \bm{C}(\tilde{{\bm{q}}})\bm{\omega}^{\text{ref}}$, where $\bm{\omega}^{\text{ref}}$ is the desired angular velocity and $\bm{C}\left(\tilde{\bm{q}}\right)$ is the rotation matrix corresponding to $\tilde{\bm{q}}$. This paper addresses the attitude stabilization problem, therefore, the desired attitude is $\bm{q}^{\text{ref}}=\begin{bmatrix} 1 & 0 & 0 & 0 \end{bmatrix}^{\top}$ and the desired angular velocity is $\bm{\omega}^{\text{ref}} = \bm{0}$. This simplifies the error quaternion and angular velocity error to $\tilde{{\bm{q}}} = \bm{q}$ and $\tilde{\bm{\omega}} = \bm{\omega}$, respectively.

After capturing a non-cooperative target, the inertia matrix of the combined body, $\bm{J}_{c}$, becomes uncertain. This uncertainty is modeled by representing the inertia matrix as:
\[
\bm{J}_{c} = \bm{J}_{c0} + \Delta\bm{J}_{c}
\]
where $\bm{J}_{c0}$ is the known nominal inertia of the combined spacecraft, and $\Delta\bm{J}_{c}$ represents the unknown inertia deviation due to the captured target. The inverse of $\bm{J}_{c}$ can be obtained as \cite{liu2024attitude}:
\[
\bm{J}_{c}^{-1} = \bm{J}_{c0}^{-1} + \Delta\bm{J}_{c}^{*}
\]
where $\Delta\bm{J}_{c}^{*}=-\left(\mathbf{I}_{3}+\bm{J}_{c0}^{-1} \Delta\bm{J}_{c}\right)^{-1} \bm{J}_{c0}^{-1} \Delta\bm{J}_{c} \bm{J}_{c0}^{-1}$. By substituting this into the Euler equation \eqref{AttDyn}, the attitude dynamics can be rewritten in terms of the nominal inertia, yielding the following error system model for control design:
\begin{subequations} \label{Un_AttDyn}
\begin{align}
\dot{{\bm{q}}} &= \bm{\Omega}({\bm{\omega}}){\bm{q}} \\
\bm{J}_{c0} \dot{{\bm{\omega}}} &= - \bm{\omega}^{\times}\bm{J}_{c0}\bm{\omega} + \bm{\tau} + \bm{d}
\end{align}
\end{subequations}
where the term $\bm{d} = 
-\bm{J}_{c0}\Delta\bm{J}_{c}^{*}\bm{\omega}^{\times}\bm{J}_{c0}\bm{\omega} - \bm{\omega}^{\times}\Delta\bm{J}_{c}^{*}\bm{\omega} - 
\bm{J}_{c0}\Delta\bm{J}_{c}^{*}\bm{\omega}^{\times}\Delta\bm{J}_{c}^{*}\bm{\omega} + \bm{J}_{c0}\Delta\bm{J}_{c}^{*}\bm{\tau} +
\left(\mathbf{I}_{3}+\bm{J}_{c0}\Delta\bm{J}_{c}^{*}\right)\bm{\tau}_d
$
represents the lumped disturbance arising from the inertial uncertainty and external disturbance.
\subsection{Control Objectives}
The primary objective is to robustly stabilize the attitude of the combined spacecraft. In particular, the control objectives are twofold:
\begin{enumerate}
    \item \textbf{Attitude Stabilization:} The controller must ensure that the state error $\begin{bmatrix}\tilde{\bm{q}}_v^{\top}& \tilde{\bm{\omega}}^{\top} \end{bmatrix}^{\top}
    \in\mathbb{R}^6$ is \emph{Uniformly Ultimately Bounded} (UUB), driving it to an arbitrarily small neighborhood of the origin, despite inertia uncertainties and target-induced disturbances.
    \item \textbf{Input Constraints:} The control torque vector $\bm{\tau} = \begin{bmatrix}\tau_1 & \tau_2 & \tau_3 \end{bmatrix}^{\top}$ must satisfy the actuator saturation limits at all times:
    \begin{equation}
        |\tau_i(t)| \leq \tau_{\max}, \quad \forall t \geq 0, \quad i\in\mathbb{I}_1^3
    \end{equation}
\end{enumerate}
where $\tau_{\max}$ is the maximum available torque for each axis.

\section{Koopman-based Model} \label{Sec_Koopman}
To achieve the control objectives defined above, a predictive control strategy is employed, whose performance, however, critically depends on the accuracy and computational efficiency of the prediction model.
To this end, we leverage the Koopman operator theory to transform the nonlinear dynamics into a linear system, which is particularly well-suited for the design of model predictive controllers.
This section is dedicated to developing such a Koopman-based model. First, we provide a brief overview of the Koopman operator theory. Then, we derive a nominal Koopman model for the nonlinear attitude dynamics and subsequently design an online update law to adapt the model to uncertainties and disturbances.
\subsection{Koopman-Operator Theory}
Consider a nonlinear autonomous dynamical system
\begin{equation} \label{sys} 
    \dot{\bm{x}} = \bm{f}(\bm{x})
\end{equation}
where $\bm{x} \in \mathbb{X} \subset \mathbb{R}^n$ denotes the state and $\bm{f}:\mathbb{X} \rightarrow \mathbb{X}$ represents the nonlinear dynamics. Let $\bm{S}_t(\bm{x}_0)$ denote the flow map of the system at time $t$. Define $\mathcal{F}$ as a space of observables, where each observable $\phi:\mathbb{X}\rightarrow\mathbb{C}$ maps the state to a measurement or feature. In this context, then the Koopman operator $\mathcal{K}^t:\mathcal{F}\rightarrow\mathcal{F}$ is defined by
\begin{equation} \label{KO}
    (\mathcal{K}^t \phi)(\bm{x}) = \phi (\bm{S}_t(\bm{x}_0))
\end{equation}

It is important to note that while the Koopman operator is linear, it is inherently infinite-dimensional, acting on the entire space of observables. In practice, however, one seeks to find a finite-dimensional invariant subspace $\mathcal{F}_{n_s}\subseteq\mathcal{F}$, spanned by a chosen set of basis functions
\[
\bm{\Phi}(\bm{x}) = \begin{bmatrix}\phi_1^\top(\bm{x}) & \phi_2^\top(\bm{x}) & \cdots & \phi_{n_s}^\top(\bm{x}) \end{bmatrix}^\top
\]
If the subspace spanned by these functions is Koopman invariant, then a finite-dimensional representation of the system can be obtained, i.e.,
\begin{equation}
    \dot{\bm{\Phi}}(\bm{x}) = \frac{\partial \bm{\Phi}}{\partial \bm{x}}\bm{f}(\bm{x})=\bm{A}\bm{\Phi}(\bm{x})
\end{equation}
with a constant matrix $\bm{A}\in\mathbb{R}^{n_s\times n_s}$.

\subsection{Physics-based Nominal Koopman Model}
Unlike data-driven methods such as EDMD, which often rely on generic polynomial or radial-basis dictionaries and consequently yield high-dimensional black-box liftings, we develop a physics-based model by exploiting the specific structure of spacecraft dynamics. The key insight is twofold: first, Euler's equation is transformed into a linear relationship by redefining the control input. Second, the inherent bilinear structure of the quaternion kinematics allows for a recursive construction of the Koopman model. The resulting model is therefore tailored to combined-spacecraft attitude dynamics, rather than an arbitrary nonlinear system, and provides a compact, physically interpretable, and computationally efficient foundation for control design.

To construct the nominal Koopman model, the inertia uncertainty is temporarily not considered, i.e., $\bm{J}_{c0} = \bm{J}_{c}$. We first define $\bm{u} = - \bm{\omega}^{\times}\bm{J}_c\bm{\omega} + \bm{\tau}$, so that \eqref{AttDyn} can be rewritten as:
\begin{equation} \label{AttDynKO} 
    \dot{{\bm{\omega}}} = \bm{\mathcal{J}}_{c0}{\bm{u}}
\end{equation}
where $\bm{\mathcal{J}}_{c0} = \bigl[\begin{array}{c|c|c} \bm{\mathcal{J}}_{c0,1} & \bm{\mathcal{J}}_{c0,2} & \bm{\mathcal{J}}_{c0,3} \\ \end{array}\bigr] = \bm{J}_{c0}^{-1}$. Next, an observable function is defined as $\bm{\phi}_1 = \dot{{\bm{q}}} = \bm{\Omega}({\bm{\omega}}){{\bm{q}} \in\mathbb{R}^4}$, with dynamics given by
\begin{equation} 
\begin{aligned}
    \dot{\bm{\phi}}_1 &= \bm{\Omega}({\bm{\omega}}){\dot{{\bm{q}}}} + \bm{\Omega}(\dot{{\bm{\omega}}}){{\bm{q}}} \\
    &= \bm{\Omega}^2({\bm{\omega}}){\bm{q}} + \bm{\Omega}(\bm{\mathcal{J}}_{c0}{\bm{u}}){\bm{q}}
\end{aligned}
\end{equation}
Let $\bm{\phi}_2 =\bm{\Omega}({\bm{\omega}})\bm{\phi}_1 = \bm{\Omega}^2({\bm{\omega}}){\bm{q}} \in\mathbb{R}^4$, so that
\begin{equation} 
    \dot{\bm{\phi}}_1 = \bm{\phi}_2  + \bm{\Omega}(\bm{\mathcal{J}}_{c0}{\bm{u}}){{\bm{q}}}
\end{equation}
Similarly, the dynamics of $\bm{\phi}_2$ is given by
\begin{equation} 
\begin{aligned}
     \dot{\bm{\phi}}_2 &= \bm{\Omega}({\bm{\omega}})\dot{\bm{\phi}}_1 + \bm{\Omega}(\bm{\mathcal{J}}_{c0}{\bm{u}})\bm{\phi}_1 \\
     &= \bm{\phi}_3 + \bm{\Omega}(\bm{\mathcal{J}}_{c0}{\bm{u}})\bm{\phi}_1 + \bm{\Omega}({\bm{\omega}})\bm{\Omega}(\bm{\mathcal{J}}_{c0}{\bm{u}}){{\bm{q}}}
\end{aligned}
\end{equation}
where $\bm{\phi}_3 = \bm{\Omega}({\bm{\omega}})\bm{\phi}_2 = \bm{\Omega}^3({\bm{\omega}}){\bm{q}}\in\mathbb{R}^4$. 
More generally, for $k \geq 2$ we define $\bm{\phi}_k = \bm{\Omega}({\bm{\omega}})\bm{\phi}_{k-1}\in\mathbb{R}^4$, with dynamics
\begin{equation} \label{dot_zk}
\dot{\bm{\phi}}_k = \bm{\phi}_{k+1} +  \sum_{i=1 }^{k}\bm{\Omega}^{i-1}({\bm{\omega}})\bm{\Omega}(\bm{\mathcal{J}}_{c0}{\bm{u}})\bm{\Omega}^{k-i}({\bm{\omega}}) {\bm{q}}
\end{equation}
\eqref{dot_zk} can be further rewritten as a control-affine form:
\begin{equation}
    \dot{\bm{\phi}}_k = \bm{\phi}_{k+1} + \bm{b}_k^c{\bm{u}}
\end{equation}
with $\bm{b}^c_k = \bigl[\begin{array}{c|c|c} \bm{b}_k^{c,1} & \bm{b}_k^{c,2} & \bm{b}_k^{c,3} \\ \end{array}\bigr] $, $\bm{b}_k^{c,n} = \sum_{i=1 }^{k}\bm{\Omega}^{i-1}({\bm{\omega}})\bm{\Omega}(\bm{\mathcal{J}}_{c0,n})\bm{\Omega}^{k-i}({\bm{\omega}}) {\bm{q}}$ for $n\in\mathbb{I}_1^3$.
By defining the lifted state vector $\bm{\mathcal{Z}} = \big[ {\bm{\omega}}^{\top} \ {\bm{q}}^{\top} \ \bm{\phi}_1^{\top} \ \bm{\phi}_2^{\top} \ \cdots \big]^{\top}$, the nominal dynamics can be represented in the infinite-dimensional linear form
\begin{equation} \label{dotZ}
    \dot{\bm{\mathcal{Z}}} = \bm{\mathcal{A}}^c\bm{\mathcal{Z}} + \bm{\mathcal{B}}^c{\bm{u}}
\end{equation}
where 
\[
    \bm{\mathcal{A}}^c = \begin{bmatrix}
    \bm{0} &\bm{0} & \bm{0} & \bm{0} & \bm{0} & \cdots \\
    \bm{0} &\bm{0} & \bm{\mathrm{I}}_4 & \bm{0} & \bm{0} & \cdots \\
    \bm{0} &\bm{0} & \bm{0} & \bm{\mathrm{I}}_4 & \bm{0} & \cdots \\
    \bm{0} &\bm{0} & \bm{0} & \bm{0} & \bm{\mathrm{I}}_4 &  \cdots \\
    \vdots & \vdots & \vdots & \vdots &  \vdots &  \ddots
    \end{bmatrix},
    \bm{\mathcal{B}}^c = \begin{bmatrix}
    \bm{\mathcal{J}}_{c0} \\ \bm{0} \\ \bm{b}^c_{1} \\ \bm{b}^c_{2} \\ \vdots
        \end{bmatrix}.
\]

Since \eqref{dotZ} is infinite-dimensional, it is not directly amenable to controller design. Therefore, we truncate the lifted state by retaining the original state variables along with the first $N_{\phi}$ observable variables, forming the finite-dimensional lifted state vector, i.e., $\bm{\Phi}({\bm{x}}) = \big[ {\bm{\omega}} \ {\bm{q}} \ \bm{\phi}_1 \ \bm{\phi}_2 \ \cdots \ \bm{\phi}_{N_{\phi}} \big]$, then \eqref{dotZ} is truncated to
\begin{equation} \label{dotZ_c}
\begin{cases}
   \dot{{\bm{z}}} = \bm{{A}}^c{{\bm{z}}} + \bm{{B}}^c{\bm{u}} \\
   {\bm{x}} = \bm{{C}}^c{ {\bm{z}} }
\end{cases}
\end{equation}
with 
\[
{\bm{z}} = \bm{\Phi}({\bm{x}}) = \begin{bmatrix}
{\bm{\omega}}^{\top} & {\bm{q}}^{\top} & \bm{\phi}_1^{\top} & \bm{\phi}_2^{\top} & \cdots & \bm{\phi}_{N_{\phi}}^{\top}
\end{bmatrix}^{\top},
\]
\[
\bm{{A}}^c = \begin{bmatrix}
    \bm{0} &\bm{0} & \bm{0} & \bm{0} & \cdots & \bm{0} \\
    \bm{0} &\bm{0} & \bm{\mathrm{I}}_4 & \bm{0} & \cdots & \bm{0} \\
    \bm{0} &\bm{0} & \bm{0} & \bm{\mathrm{I}}_4 & \cdots & \bm{0} \\
    \vdots &\vdots & \vdots & \vdots & \ddots &  \vdots \\
    \bm{0} &\bm{0} & \bm{0} & \bm{0} & \cdots &  \bm{\mathrm{I}}_4 \\
    \bm{0} &\bm{0} & \bm{0} & \bm{0} & \cdots &  \bm{0}
    \end{bmatrix},
    \bm{{B}}^c = \begin{bmatrix}
   \bm{\mathcal{J}}_{c0} \\ \bm{0} \\ \bm{b}^c_1 \\ \bm{b}^c_2 \\ \vdots \\ \bm{b}^c_{n}
        \end{bmatrix}, 
\]
and $\bm{{C}}^c=\begin{bmatrix} \textbf{I}_7 & \textbf{0}_{7 \times 4N_{\phi}} \end{bmatrix}$. This yields an approximate $(4N_{\phi}+7)$-dimensional model with 3-dimensional control inputs.

After deriving the finite-dimensional Koopman model, we now turn to its theoretical verification. To provide the necessary justification for using this model in control design, we first establish its fundamental properties, which ensure that key characteristics like stability are preserved from the original nonlinear system.

\begin{assumption} \label{ASS1}
There exists a compact forward-invariant set $\mathbb{X} \subset \mathbb{R}^n$ containing all feasible closed-loop trajectories generated by the admissible initial conditions and bounded control inputs considered in this paper.
\end{assumption}
\begin{assumption} \label{PROP2}
The physics-based lifting function $\bm{\Phi}(\bm{x})$, which is composed of polynomial terms of the state $\bm{x}$, is Lipschitz continuous on the compact set $\mathbb{X}$ defined in Assumption~\ref{ASS1}.
\end{assumption}
With these established properties, we can therefore state the following lemma.
\begin{lemma}
The Lyapunov-stable equilibrium $\bm{x}_e=[0 \ 0 \ 0 \ 1 \ 0 \ 0 \ 0]^{\top}$ of the original system \eqref{Un_AttDyn} is also a Lyapunov-stable equilibrium for the Koopman model \eqref{dotZ_c}, i.e., if
\begin{equation} 
\begin{aligned}
&\exists \eta > 0, \forall \epsilon>0, \exists \delta(\varepsilon) > 0,\\
&\text{s.t.}\ \Vert \bm{x}(0)-\bm{x}_e \Vert < \delta \Rightarrow  \Vert \bm{x}(t)-\bm{x}_e \Vert < \epsilon, \forall t\geq0 \\
\end{aligned}
\end{equation}
Then it holds
\begin{equation} 
\begin{aligned}
&\exists \eta_{\Phi} > 0, \ \forall \epsilon_{\Phi}>0,  \ \exists \delta_{\Phi}(\varepsilon_{\Phi}) > 0,\\
&\text{s.t.}\ \quad\Vert \bm{\Phi}\bigl(\bm{x}(0)\bigr)-\bm{\Phi}(\bm{x}_e) \Vert < \delta_{\Phi} \\
&\ \quad \Rightarrow  \Vert \bm{\Phi}\bigl(\bm{x}(t)\bigr)-\bm{\Phi}(\bm{x}_e) \Vert < \epsilon_{\Phi}, \forall t\geq0 \\
\end{aligned}
\end{equation}
\end{lemma}
\begin{IEEEproof} 
See \cite[Pro. 1]{10091950} for the proof.
\end{IEEEproof}

\begin{lemma}
The linear Koopman system \eqref{dotZ_c} is controllable.
\end{lemma}

\begin{IEEEproof}
Inspired by \cite[Ch. 9]{Chen2023}, we introduce a new control input variable ${\bm{v}}$ and define it via the relation
\begin{equation} \label{v}
\bm{{B}}^c({\bm{q}}, {\bm{\omega}}){\bm{u}} = \bar{\bm{B}}^c{\bm{v}}
\end{equation}
where
\[
\bm{v} =  \begin{bmatrix}\bm{v}_1 \\ \bm{v}_2 \\ \vdots \\ \bm{v}_{n+1}\end{bmatrix} 
= \begin{bmatrix}
    \bm{\mathcal{J}}_{c0}\bm{u} \\
    \bm{b}^c_1\bm{u} \\
    \vdots \\
   \bm{b}^c_n\bm{u}
\end{bmatrix}, \quad
\bar{\bm{B}}^c = \begin{bmatrix}
    \begin{array}{cc} \textbf{I}_{3} &  \textbf{0}_{3 \times 4N_{\phi}} \end{array} \\
   \textbf{0}_{4 \times (4N_{\phi}+3)} \\
     \begin{array}{cc}  \textbf{0}_{4N_{\phi} \times 3} &  \textbf{I}_{4N_{\phi}} \end{array}
\end{bmatrix}.
\]
With this transformation, system \eqref{dotZ_c} can be described as a linear time-invariant (LTI) form:
\begin{equation} \label{dotZ_c_LTI}
    \dot{{\bm{z}}} = \bm{{A}}^c{\bm{z}} + \bar{\bm{B}}^c{\bm{v}}
\end{equation}

It's important to note that
\[
\operatorname{rank}\!\left[\,\bm{B}^c \quad \bar{\bm{B}}^c\,\right] = \operatorname{rank}\left(\bm{B}^c\right)
\]
holds for all admissible ${\bm{q}}$ and ${\bm{\omega}}$, which guarantees that for every admissible $\bm{v}$ there exists at least one $\bm{u}$ satisfying \eqref{v},
thus the mapping $\bm{v} \mapsto \bm{u}$ is surjective, ensuring that any control input $\bm{v}$ can be implemented via an appropriate choice of $\bm{u}$ in the original system.

Moreover, it can be computed that the controllability matrix $\bm{M}_c = \left [ \bar{\bm{B}}^c \  \bm{A}^c\bar{\bm{B}}^c \ \dots \ (\bm{A}^c)^{4N_{\phi}+7}\bar{\bm{B}}^c \right ]$ is full-column rank, which implies that system \eqref{dotZ_c_LTI} is controllable. Since the mapping $\bm{v} \mapsto \bm{u}$ is onto, the controllability of the transformed LTI system directly ensures the controllability of the original system \eqref{dotZ_c}. For brevity, the detailed computation of the rank of \(\bm{M}_c\) is omitted.
\end{IEEEproof}

For the subsequent design of the MPC strategy, it is necessary to convert system \eqref{dotZ_c} into a discrete-time model.
Assuming the control input is generated by an ideal zero-order hold (ZOH), i.e., $\bm{u}(t):=\bm{u}_k$ for $t \in \left [ kT_s,(k+1)T_s \right ]$.
The existence of a unique discrete-time model is guaranteed by the following proposition.
\begin{proposition}
The system \eqref{dotZ_c} under ZOH admits a unique, exact discrete-time map $\bm{F}_{T_s}: \mathbb{R}^{4N_{\phi}+7} \times \mathbb{R}^3 \to \mathbb{R}^{4N_{\phi}+7}$ such that $\bm{z}_{k+1} = \bm{F}_{T_s}(\bm{z}_k, \bm{u}_k)$.
\end{proposition}
\begin{IEEEproof}
The system matrices $\bm{A}^c$ and $\bm{C}^c$ are constant, and $\bm{B}^c$ is locally Lipschitz continuous.
By the Picard–Lindelöf theorem \cite[Thm. 3.1]{khalil2002nonlinear}, this guarantees the local existence and uniqueness of a solution to \eqref{dotZ_c} starting from the initial time $t = kT_s$. For physical systems such as spacecraft, solutions do not exhibit finite-time escape under bounded inputs, ensuring the solution exists over the entire interval $[kT_s, (k+1)T_s]$. The state at time $(k+1)T_s$ is thus a unique function of the initial state $\bm{z}_k$ and input $\bm{u}_k$.
\end{IEEEproof}
For practical implementation, the continuous-time system can be discretized as:
\begin{equation} \label{koop_dis}
\begin{cases}
   {\bm{z}}_{k+1} = \bm{A}{\bm{z}}_{k} + \bm{{B}}\bm{u}_k \\
   {\bm{x}}_{k} = \bm{{C}}{\bm{z}}_{k}
\end{cases}
\end{equation}
where $\bm{{A}} =e^{\bm{{A}}^c{T_s}}$, $\bm{{B}} = \left(\int_{0}^{T_s} e^{\bm{A}^c \tau}\, \mathrm{d}\tau\right)\bm{B}^c$, $\bm{{C}} = \bm{{C}}^c$, $T_s\in\mathbb{R}_0^{+}$ denotes the sampling time.

\begin{remark}
From a theoretical perspective, the proposed Koopman modeling method is more efficient in terms of on-board computational resources compared to existing alternatives. 
Specifically, for the same number of included dynamic terms $\left( N_{\phi} \right)$, our model is $5N_{\phi}+2$ dimensions more compact than methods that construct lifting functions from DCM dynamics (e.g., using $\dot{\bm{R}} = \bm{R}\bm{\omega}^{\times}$) \cite{Chen2023}.
Similarly, data-driven EDMD methods \cite{korda2018linear} often rely on a large dictionary of basis functions, where achieving comparable accuracy can require the model's dimension to be in the tens or even hundreds. The subsequent simulation results will further demonstrate this computational advantage.
\end{remark}

\subsection{Recursive Online Update Law}

Although the offline nominal model guarantees a global approximation of the attitude dynamics, model mismatch due to a lack of prior knowledge of inertia or external noise remains inevitable in practical scenarios. As will be demonstrated in (Section~\ref{Section-V-C}), these mismatches lead to significant performance degradation if left uncompensated. To compensate for such mismatches, an online adaptive Koopman method is employed. This method recursively updates the nominal model by minimizing the prediction error between the actual lifted state and the predicted lifted state.

Consider a control-affine system with uncertainty:
\begin{equation} 
    \dot{\bm{x}} = \bm{f}(\bm{x}) + \tilde{\bm{f}}(\bm{x}) + \bigl(\bm{g}(\bm{x}) + \tilde{\bm{g}}(\bm{x})\bigr)\bm{u}
\end{equation}
where ${\bm{f}}$ and ${\bm{g}}$ denote the unknown or uncertain dynamics. Assume that $\frac{\partial \bm{\Phi}}{\partial \bm{x}}{\bm{f}}(\bm{x})\in \text{span}\left \lbrace \phi_1,\phi_2,\cdots,\phi_{n_s} \right \rbrace$ and $\frac{\partial \bm{\Phi}}{\partial \bm{x}}{\bm{g}}(\bm{x})\in \text{span}\left \lbrace \phi_1,\phi_2,\cdots,\phi_{n_s} \right \rbrace$, the lifted representation can be obtained as:
\begin{equation} 
\dot{\bm{\Phi}}(\bm{x}) = (\bm{A} + \Delta\bm{A})\bm{\Phi}(\bm{x}) + (\bm{B}(\bm{x}) + \Delta\bm{B}(\bm{x}))\bm{u}
\end{equation}
where
\[
\Delta\bm{A}\bm{\Phi}(\bm{x})=\frac{\partial \bm{\Phi}}{\partial \bm{x}}{\bm{f}}(\bm{x}) \qquad
\Delta\bm{B}(\bm{x})=\frac{\partial \bm{\Phi}}{\partial \bm{x}}{\bm{g}}(\bm{x})
\]

Based on the previous section, the discrete-time Koopman representation of \eqref{Un_AttDyn} can be obtained as
\begin{equation} \label{un_koop} 
    \bm{z}_{k+1} = \bm{A}\bm{z}_{k} + \bigl(\bm{B} + \Delta\bm{B}\bigr){\bm{u}}_{k}
\end{equation}
Here, $\bm{z}_{k+1}$ is obtained by lifting the true state, i.e., $\bm{z}_{k+1}=\bm{\Phi}(\bm{x}_{k+1})$, while the nominal model \eqref{koop_dis} produces a state estimate $\hat{\bm{z}}_{k+1}$.

By defining the prediction error $\Delta \bm{z}_k = \bm{z}_k - \hat{\bm{z}}_k$ at time $k$, the uncertainty matrix $\Delta\bm{B}$ can be identified by the following optimization problem:
\begin{equation} \label{opt_delta_B}
\min_{\Delta{\bm{B}}^*}\left \Vert \Delta \bm{z}_{k} - \Delta{\bm{B}}\bm{u}_k\right \Vert_2
\end{equation}

Typically the uncertainty matrix $\Delta\bm{B}$ can be identified by $\Delta\hat{\bm{B}}_k = \Delta\bm{\mathcal{Z}}_k  \bm{\mathcal{U}}_k ^{\dagger}$, where $\Delta\bm{\mathcal{Z}}_k$ and $\bm{\mathcal{U}}_k$ denote the column-wise datasets of $\Delta\bm{z}$ and $\bm{u}$, respectively. In practice, however, such a least squares approach may suffer from issues such as overfitting and numerical instability, especially if the dataset is noisy or limited in size.

Hence, we introduce an adaptive Koopman model which employs a modified recursive algorithm to estimate $\Delta\bm{B}$ online. 
To enhance the robustness of the identification process against measurement noise and transient variations, the data were retained in memory through a sliding window over the most recent $N_w$ time steps, forming the dataset $\mathcal{D}_k =\left \lbrace  {\bm{\mathcal{U}}_k, \Delta\bm{\mathcal{Z}}_k} \right \rbrace$, where:
\[
\begin{split}
\bm{\mathcal{U}}_k = \left[ {\bm{u}}_{k-N_w} \ {\bm{u}}_{k-N_w+1} \ \cdots \ {\bm{u}}_{k-1}\right],\\
\Delta\bm{\mathcal{Z}}_k = \left[ \Delta\bm{z}_{k-N_w+1} \ \Delta\bm{z}_{k-N_w+2} \ \cdots \ \Delta\bm{z}_{k}\right].
\end{split}
\]

To promote sparsity and mitigate overfitting, particularly with noisy or limited data, an $\ell{1}$-regularized exponentially-weighted objective was adopted:
\begin{equation}
\begin{aligned}
    \mathcal{L}_k( \Delta\bm{B} ) &= \frac{1}{N_w}\sum_{i=k-N_w}^{k-1}\lambda^{k-i}\left \Vert \Delta \bm{z}_{i+1} - \Delta{\bm{B}}\bm{u}_{i} \right \Vert_2^2+\gamma\Vert \Delta\bm{B} \Vert_1 \\
    &= \frac{1}{N_w}\left \Vert  \Delta\bm{\mathcal{Z}}_k - \Delta\bm{B}\bm{\mathcal{U}}_k \bm{\Lambda}^{\frac{1}{2}}\right \Vert_2^2 + \gamma\Vert \Delta\bm{B} \Vert_1
\end{aligned}
\end{equation}
where $0<\lambda<1$ denote a forgetting factor that exponentially discounts older data, $\bm{\Lambda}=\text{diag}(\lambda^{N_w-1},\lambda^{N_w-2},\dots,\lambda)$, and $\gamma > 0$ induces sparsity. At each sampling instant $k$, the gradient of $\mathcal{L}_k( \Delta\bm{B} )$ can be calculated as:
\begin{equation}
\nabla_{\Delta\bm{B}}\mathcal{L}_k = -\frac{2}{N_w} \left (  \Delta\bm{\mathcal{Z}}_k - \Delta\bm{B}\bm{\mathcal{U}}_k  \right )\bm{\Lambda}\bm{\mathcal{U}}_k^{\top} + \gamma \operatorname{sign}(\Delta \bm{B})
\end{equation}
To improve numerical stability and convergence speed, the Adam \cite{kingma2017adammethodstochasticoptimization} algorithm is chosen as the gradient descent algorithm:
\begin{equation} \label{ADAM}
\begin{split}
\bm{m}_{k} &= \beta_1\bm{m}_{k-1} + (1 - \beta_1)\nabla_{\Delta\bm{B}}\mathcal{L}_k,\\
\bm{v}_{k} &= \beta_2\bm{v}_{k-1} + (1 - \beta_2)\nabla_{\Delta\bm{B}}\mathcal{L}_k\odot\nabla_{\Delta\bm{B}}\mathcal{L}_k,\\
\hat{\bm{m}}_{k} &= \frac{\bm{m}_k}{1 - \beta_1^k}, \quad \hat{\bm{v}}_{k}=\frac{\bm{v}_k}{1 - \beta_2^k},
\\
\Delta \bm{B}_{k}^{*}&=\Delta \bm{B}_{k-1}-\rho \frac{\hat{\bm{m}}_{k}}{\sqrt{\hat{\bm{v}}_{k}}+\hat{\varepsilon}},
\end{split}
\end{equation}
where $\rho$ denotes the learning rate, $\bm{m}_{k}$ and $\bm{v}_{k}$ denote the estimates of the first and second moment of the gradients, respectively. $\hat{\bm{m}}_{k}$ and $\hat{\bm{v}}_{k}$ are their bias-corrected counterparts. The hyperparameters $\beta_1$ and $\beta_2$ denote the exponential decay rates for estimating moments, with their default values commonly set to 0.9 and 0.999. $\hat{\varepsilon}$ is a small constant, usually taken as $1e^{-8}$.
This yields a gradient-driven, memory-efficient algorithm whose computational complexity grows linearly with the size of $\Delta\bm{B}$, making it suitable for real-time, high-dimensional Koopman operator refinement under streaming data.
\begin{remark}
The gradient-based update does not explicitly enforce \emph{persistent excitation} (PE), which is typically required for parameter convergence. Here, the goal of the online adaptation is not exact parameter convergence, but local model accuracy for closed-loop control. In practice, excitation mainly arises during transient operation from the tracking task, constraints, and disturbances, leading to time-varying inputs $\bm{u}_k$. As the system approaches steady state, the update naturally weakens toward the nominal model, while regularization and forgetting mechanisms help mitigate parameter drift in low-excitation regions.
\end{remark}

\begin{remark}
The adaptive matrix $\Delta \bm{B}$ provides a control-oriented, lumped representation of uncertainty affecting the input-dependent part of the lifted model, rather than a physical decomposition of all uncertainty sources. In the proposed framework, it is primarily intended to compensate for inertia-induced modeling errors and bounded target-generated torques that can be locally represented as an equivalent input-channel disturbance, rather than exactly reconstructing arbitrary unmatched additive disturbances. If the dominant uncertainty is not input-matched, the current $\Delta \bm{B}$-only update can be extended to a more general adaptive lifted model involving $\Delta \bm{A}$, $\Delta \bm{B}$, and an additive disturbance channel. Practical stability can be maintained provided that the residual modeling error is bounded and incorporated into the admissible prediction-error/disturbance channel, and does not violate the state/input constraints or the robustness margin of the terminal set. Additional disturbance compensation or offset-free MPC mechanisms may be needed to preserve steady-state accuracy.
\end{remark}

\section{Koopman Operator-based MPC Design} \label{Sec_Controller}
\subsection{Model Predictive Controller Design}
In this subsection, we propose an MPC formulation based on the Koopman linear model and an interpolated initial state. The MPC optimization problem is formulated at each time step $k$ as follows \cite{de2024koopman}:
\begin{subequations} \label{MPC}
\begin{align}
\min_{\bm{U}^*_k,\mu^*_k}  J &= \min_{\bm{U}^*_k,\mu^*_k}\sum_{i=0}^{N-1} l \bigl( {\bm{z}}_{i|k}, {\bm{u}}_{i|k}\bigr) + V_F \bigl( {\bm{z}}_{N|k} \bigr) + l_r( {\bm{z}}_{0|k} )\\
\text{s.t.} \
& \hat{\bm{x}}_{k} = \bm{C}\hat{\bm{z}}_{k} \\
& {\bm{z}}_{0|k} = {\bm{z}}_k = (1-\mu_k)\bm{\Phi}({\bm{x}}_{k}) + \mu_k{\bm{z}}^*_{1|k-1} \label{MPC_cons_ISS}\\
& \hat{{\bm{z}}}_{i+1|k} = \bm{A} {\bm{z}}_{i|k} + \bm{B} {\bm{u}}_{i|k}, \ \forall i \in \mathbb{N}_{\leq N-1} \label{MPC_cons_tough}\\
& \bm{C}\hat{\bm{z}}_{i|k} \in \mathbb{X}, \ \forall i \in \mathbb{N}_{\leq N} \\
& {\bm{u}}_{i|k} \in \mathbb{U}, \ \forall i \in \mathbb{N}_{\leq N-1} \\
& {\bm{z}}_{N|k} \in \mathbb{Z}_F 
\end{align}
\end{subequations}
Here $l( \cdot )$ represents the stage cost function,  $V_F( \cdot )$ represents the terminal cost function, $l_r( \cdot )$ is a regularization term penalizing the difference between the interpolated initial state and the Koopman lifted state. In this work, the following choices are adopted:
\begin{subequations}
\begin{align}
l \bigl( \bm{z}_{i|k}, {\bm{u}}_{i|k}\bigr) &= \sum_{i=0}^{N-1}\bigl\Vert \bm{C}\hat{{\bm{z}}}_{i|k}\bigr\Vert^2_{\bm{Q}} + \Vert {\bm{u}}_{i|k} \Vert^2_{\bm{R}} \label{stage_cost}\\
V_F \bigl( {\bm{z}}_{N|k}\bigr) &= \bigl\Vert \hat{{\bm{z}}}_{N|k} \bigr\Vert^2_{\bm{P}} \label{ter_cost} \\
l_r( \bm{z}_{0|k} ) &= \gamma_e\Vert \bm{\Phi}({\bm{x}}_{k}) - {\bm{z}}_{0|k}\Vert_2^2 \label{reg_cost}
\end{align}
\end{subequations}
with weighting matrices $\bm{Q}\succ0$, $\bm{R}\succ0$, $\bm{P}\succ0$ and a regularization coefficient $\gamma_e > 0$.

To solve this optimization problem efficiently, we convert it into a standard QP. First, define the optimization variables that consist of the control input sequence and the interpolation parameter:
\begin{equation}
\begin{aligned}
\bm{v}_k &=\left [ {\bm{u}}^{\top}_{0|k} \ \ {\bm{u}}^{\top}_{1|k} \ \ \cdots \ \ {\bm{u}}^{\top}_{N-1|k} \ \ \mu_k\right ]^{\top} \\&= \left[ {\bm{U}}^{\top}_{k} \ \ \mu_k\right]^{\top}\in \mathbb{R}^{3N+1}
\end{aligned}
\end{equation}

Next, by defining the error term $\bm{e}_k = \bm{\Phi}({\bm{x}}_{k}) - {\bm{z}}_{1|k-1}^{*}$, the interpolated initial state \eqref{MPC_cons_ISS} can be rewritten as:
\begin{equation} \label{z0k}
{\bm{z}}_{0|k} = \bm{\Phi}({\bm{x}}_{k}) + \mu_k\bm{e}_k
\end{equation}
The predicted state trajectory $\bm{Z}_k=\left [ {\bm{z}}^{\top}_{0|k} \ {\bm{z}}^{\top}_{1|k} \ \cdots \ {\bm{z}}^{\top}_{N|k}\right ]^{\top}$ can be expressed in a compact matrix form using the system dynamics \eqref{MPC_cons_tough}:
\begin{equation}
\bm{Z}_k=\bm{\mathcal{A}}{\bm{z}}_{0|k} + \bm{\mathcal{B}}\bm{U}_k
\end{equation}
where
\[
\bm{\mathcal{A}}:={\scriptsize \begin{bmatrix}\mathbf{I} \\ \bm{A} \\ \bm{A}^2 \\ \vdots \\ \bm{A}^N \end{bmatrix}},
\bm{\mathcal{B}}:={\scriptsize \begin{bmatrix}\mathbf{0} & \mathbf{0} & \cdots & \mathbf{0} \\ 
\bm{B} & \mathbf{0} & \cdots & \mathbf{0} \\ 
\bm{A}\bm{B}  & \bm{B} & \cdots & \mathbf{0}\\ 
\vdots  & \vdots & \ddots & \vdots\\ 
\bm{A}^{N-1}\bm{B}  & \bm{A}^{N-1}\bm{B} & \cdots & \bm{B} \end{bmatrix}}.
\]
By substituting \eqref{z0k} into the trajectory equation, the state trajectory $\bm{Z}_k$ and control sequence $\bm{U}_k$ can be rewritten as functions of the optimization variable $\bm{v}_k$:
\begin{equation}
\bm{Z}_k= \bm{S}_v\bm{v}_k + \bm{S}_c, \quad \bm{U}_k= \bm{S}_u\bm{v}_k.
\end{equation}
with $\bm{S}_v = \begin{bmatrix}\bm{\mathcal{B}} & \bm{\mathcal{A}}\bm{e}_k \end{bmatrix}$, $\bm{S}_c = \bm{\mathcal{A}}\bm{\Phi}(\bm{x}_k)$, $\bm{S}_u = \begin{bmatrix} \textbf{I}_{3N} & \bm{0} \end{bmatrix}$.
Finally, we formulate the cost function \eqref{MPC} in terms of $\bm{v}_k$. Note that the regularization term \eqref{reg_cost} can be rewritten as $l_r = \lambda\mu_k^2(\bm{e}_k^{\top}\bm{e}_k)$. This allows the cost function to be expressed in the standard QP form:
\begin{equation} \label{QP}
J = \frac{1}{2}\bm{v}_k^{\top}\bm{H}\bm{v}_k + \bm{f}^{\top}\bm{v}_k + \text{constant}
\end{equation}
where the Hessian matrix $\bm{H}$ and the linear term $\bm{f}$ are given by
\[
\bm{H}=2\left ( \bm{S}_v^{\top}\bar{\bm{Q}}\bm{S}_v + \bm{S}_u^{\top}\bar{\bm{R}}\bm{S}_u + \bm{H}_\text{{reg}}\right ),
\bm{f} =  2\bm{S}_c^{\top}\bar{\bm{Q}}\bm{S}_v.
\]
Here, $\bar{\bm{Q}} = \text{diag}( \bm{C}^{\top}\bm{Q}\bm{C},\dots,\bm{C}^{\top}\bm{Q}\bm{C},\bm{P})$, $\bar{\bm{R}} = \text{diag}(\bm{R},\dots,\bm{R} )$, $\bm{H}_{\mathrm{reg}}$ is a zero matrix except for its last diagonal element $\lambda\bm{e}^{\top}_k\bm{e}_k$. This QP problem can be efficiently solved at each sampling time using standard numerical solvers to obtain the optimal control sequence and interpolation parameter.

\subsection{Terminal ingredients design}
In order to ensure the stability and recursive feasibility of the proposed MPC, it is critical to design appropriate terminal ingredients. These ingredients consist of the terminal cost function, terminal set, and terminal control gain, which together ensure that if the state reaches the terminal region, the subsequent closed-loop trajectory is guaranteed to converge asymptotically toward the equilibrium point. 

For simplicity, the uncertain Koopman model \eqref{un_koop} is represented as
\begin{equation}
{\bm{z}}_{k+1} = \bm{h}\bigl({\bm{z}}_{k}, {\bm{u}}_{k} \bigr)
\end{equation}
The following basic assumption regarding the terminal ingredients should be satisfied to facilitate the analysis.
\begin{assumption} 
Define $V_F(\cdot)$ as the terminal cost function, $\mathbb{Z}_F$ as the terminal set, and $\bm{F}$ as the terminal control gain. The following conditions are assumed to hold:
\begin{enumerate}
    \item \textbf{Terminal Set:} 
    \begin{subequations}
        \begin{align}
        \bm{h}\bigl( \mathbb{Z}_F,\bm{F}\mathbb{Z}_F \bigr) &\subseteq \mathbb{Z}_F, \\
        \bm{C}\,\mathbb{Z}_F &\subseteq \mathbb{X}, \\
        \bm{F}\,\mathbb{X} &\subseteq \mathbb{U}.
        \end{align}
    \end{subequations}
    \item \textbf{Terminal Cost:}
    \begin{subequations}
    \begin{align}
    V_F\bigl( {\bm{z}} \bigr) \geq 0,\quad V_F\bigl( 0 \bigr) = 0\\
        \begin{split} \label{Vf2}
        V_F\bigl( \bm{h}({\bm{z}},\bm{F}{\bm{z}}) \bigr)  \leq V_F\bigl( {\bm{z}} \bigr) - &l({\bm{z}},\bm{F}{\bm{z}}), \\ &\forall {\bm{z}}\in\mathbb{Z}_{F} .
        \end{split}
    \end{align}
    \end{subequations}
    \end{enumerate}
\end{assumption}

Under the above assumption, the stage cost function and terminal cost function are define as \eqref{stage_cost} and \eqref{ter_cost}, respectively, and the terminal set is defined as $\mathbb{Z}_F:= \lbrace {\bm{z}}: V_F({\bm{z}}) = \Vert {\bm{z}} \Vert_P^2 \leq \sigma \rbrace$ with a threshold $\sigma>0$.

In most studies, to guarantee the decrease condition \eqref{Vf2} of the terminal cost, in the nominal case (ignoring the disturbance), the
terminal weighting matrix $\bm{P}$ is obtained by solving the following Lyapunov inequality:
\[
\left(\boldsymbol{A}+\boldsymbol{B} \boldsymbol{F}\right)^{\top} \boldsymbol{P}\left(\boldsymbol{A}+\boldsymbol{B} \boldsymbol{F}\right)-\boldsymbol{P}+\boldsymbol{Q}+\boldsymbol{F}^{\top} \boldsymbol{R F} \leq 0
\]
However, in our study, an adaptive Koopman framework is adopted, implying that the system matrix $\bm{B}$ is updated online, further leading to the online computation of $\bm{P}$. To reduce the burden of online computing, the uncertain model \eqref{un_koop} is again recast as:
\begin{equation}
\begin{aligned}
    {\bm{z}}_{k+1} &= \bm{A}{\bm{z}}_{k} + \bm{B}_0{\bm{u}}_{k} + \bm{B}_1{\bm{u}}_{k} +\Delta\bm{B}{\bm{u}}_{k} \\
    &= \bm{A}{\bm{z}}_{k} + \bm{B}_0{\bm{u}}_{k} + \tilde{\Delta}{\bm{z}}_{k}
\end{aligned}
\end{equation}
with a constant matrix $\bm{B}_0:=\bm{B}_0({\bm{q}}_0,{\bm{\omega}}_0)$ constructed from known initial state, a state-dependent term $\bm{B}_1:=\bm{B}_1({\bm{q}}_k,{\bm{\omega}}_k)$ satisfying a boundedness condition (i.e., $\bigl\Vert \bm{B}_1 {\bm{u}}_{k} \bigr\Vert \leq \epsilon \bigl\Vert {\bm{u}}_{k} \bigr\Vert$),  and a extended perturbation $\tilde{\Delta}{\bm{z}}_{k} = \bm{B}_1{\bm{u}}_{k} +\Delta\bm{B}{\bm{u}}_{k}$. 

Taking $\tilde{\Delta}{\bm{z}}_{k}$ into account, the overall decrease condition must be modified as follows:
\begin{equation} \label{P_ineq_all}
\begin{split} 
    (\bm{A}_{F}{\bm{z}}_k + \tilde{\Delta}{\bm{z}}_{k})^{\top}\bm{P}(\bm{A}_{F}{\bm{z}}_k \ + &\ \tilde{\Delta}{\bm{z}}_{k}) \\
    -{\bm{z}}_k^{\top}\bm{P}{\bm{z}}_k + {\bm{z}}_k^{\top}\bm{Q}{\bm{z}}_k + &{\bm{z}}_k^{\top}\bm{F}^{\top}\bm{R}\bm{F}{\bm{z}}_k \leq 0 
\end{split}
\end{equation}
where $\bm{A}_{F}:=\bm{A}+\bm{B}_0\bm{F}$. 
By introducing a to-be-designed scalar $\kappa \in (0,1)$, \eqref{P_ineq_all} is equivalent to:
\begin{equation} \label{P_ineq}
\bm{A}_F^{\top} \bm{P} \bm{A}_F-\left(1-\kappa\right) \bm{P}+ \bm{Q} +\bm{F}^{\top} \bm{R} \bm{F} \leq 0 
\end{equation}
with the extended perturbation satisfying
\begin{equation}
    \tilde{\Delta}\bm{z}_{k}^{\top}\bm{P}\tilde{\Delta}\bm{z}_{k} + 2 {\bm{z}}_k^{\top}\bm{A}^{\top}_{F}\bm{P}\tilde{\Delta}\bm{z}_{k} \leq \kappa {\bm{z}}_k^{\top}\bm{P}{\bm{z}}_k
\end{equation}
for all ${\bm{z}}_k \in \mathbb{Z}_F$.

By focusing on \eqref{P_ineq} and applying the Schur complement, a linear matrix inequality (LMI) can be derived as
\begin{equation} \label{LMI}
\left[\begin{array}{cccc}\left(1-\kappa\right) \bm{O} & \left(\bm{A}\bm{O} + \bm{B}_0\bm{Y}\right)^{\top} & \bm{O} & \bm{Y}^{\top} \\ \bm{A}\bm{O} + \bm{B}_0\bm{Y} & \bm{O} & \bm{0} & \bm{0} \\ \bm{O} & \bm{0} & \frac{1}{\sigma}\bm{Q}^{-1} & \bm{0} \\ \bm{Y} & \bm{0} & \bm{0} & \frac{1}{\sigma}\bm{R}^{-1}\end{array}\right] \succeq 0
\end{equation}
with the substitutions $\bm{O}=\sigma\bm{P}^{-1}$ and $\bm{Y}=\bm{F}\bm{O}$.
The design of $\bm{P}$ and $\bm{K}$ is then formulated as the following optimization problem \cite{lazar2018computation}:
\begin{equation} \label{opt_P}
\begin{aligned}
\min_{\sigma,\bm{O},\bm{Y}} \quad & \sigma \\
\text{s.t.} \quad 
& \eqref{LMI},\\
& \bm{H}_x\bm{O}\bm{H}_x^{\top} \leq\bm{h}_x, \\
& \begin{bmatrix}
    1 & \bm{H}_u\bm{Y} \\
    \bm{Y}^{\top}\bm{H}_u^{\top} & \bm{O}
\end{bmatrix} \geq\bm{0}.
\end{aligned}
\end{equation}
Once the feasible solution is obtained, $\bm{P}$ and $\bm{K}$ can be obtained by $\bm{P} = \sigma\bm{O}^{-1}$ and $\bm{F} =\bm{Y}\bm{O}^{-1}$.

It is worth noting that the proposed terminal-set construction is designed with a robustness margin, since fixed terminal ingredients are employed to retain tractable recursive-feasibility and stability guarantees under online model adaptation. As a result, a certain degree of conservatism may be introduced.  Alternative terminal constructions with reduced conservatism have been studied in the literature, e.g., \cite{chai2021dual,zhang2022robust}.

The overall adaptive Koopman-based MPC algorithm is summarized in Algorithm \ref{alg_1}.
\begin{algorithm}[H]
\renewcommand{\algorithmicrequire}{\textbf{Initialization:}}
\renewcommand{\algorithmicensure}{\textbf{Offline Computation:}}
\caption{Adaptive Koopman-based MPC for Uncertain Attitude Dynamic}\label{alg:alg1}
\label{alg_1}
\begin{algorithmic}[1]
\REQUIRE Choose $k=0$, $\bm{J}_{c0}$, $\bm{Q}$, $\bm{R}$, $T_{\text{s}}$, $N$, $N_{\phi}$, $N_w$, \\
    $\kappa$, $\gamma$, $\gamma_e$, $\rho$, $\lambda$.
\ENSURE 
\STATE Calculate nominal Koopman model matrix $\bm{A}$ and $\bm{B}$;
\STATE Set $\bm{B}[0] = \bm{B}(\bm{q}_0,\bm{\omega}_0)$ when $k=0$;
\STATE Compute $\bm{P}$ by solving \eqref{opt_P}.
\renewcommand{\algorithmicensure}{\textbf{Online Learning Phase:}}
\ENSURE 
\WHILE{$k \geq 0$}
\IF {$k < N_w$} 
\STATE $\Delta\bm{B} = \bm{0}$
\ELSE
\STATE Generate online learning dataset $\bm{U}^{\text{data}}_k$ and $\Delta\bm{Z}^{\text{data}}_k$
\STATE Compute $\Delta{\bm{B}}$ using \eqref{ADAM}
\ENDIF
\STATE $\bm{B} \leftarrow \bm{B} + \Delta\bm{B}$;
\STATE Compute $\bm{U}^*_k$ by solving \eqref{MPC};
\STATE Apply the first control input $\bm{u}_k = \bm{u}^*_{0|k}$;
\STATE $k \leftarrow k+1$
\ENDWHILE
\end{algorithmic}
\label{alg1}
\end{algorithm}

\subsection{Recursive Feasibility and Stability Analysis}
This section provides a formal analysis of the controller's essential theoretical properties. We will first prove that the MPC optimization problem is recursively feasible, ensuring persistent operation. Subsequently, we will demonstrate that the closed-loop system is input-to-state stable (ISS) with respect to the Koopman model prediction errors, guaranteeing robust performance in the presence of uncertainty.
\begin{theorem} \label{theorem_Rec} 
The Koopman-based MPC problem \eqref{MPC} is recursively feasible.
\end{theorem}
\begin{IEEEproof} 
Suppose that at time $k$, there exists an optimal control input sequence $\bm{U}^{*}_k = \bigl[ \bm{u}^{*}_{0|k},\bm{u}^{*}_{1|k},\cdots, \bm{u}^{*}_{N-1|k}\bigr]$, and optimal terminal state satisfies
$V_F\bigl(\bm{z}^*_{N|k}\bigr) = \bigl\Vert\bm{z}^*_{N|k} \bigr\Vert_P^2\leq\sigma$.
At time $k+1$, a candidate sub-optimal solution can be given by:
\begin{equation}
\begin{aligned}
\bm{U}_{k+1} &= \bigl[ \bm{u}_{0|k+1},\bm{u}_{1|k+1},\cdots, \bm{u}_{N-2|k+1},\bm{u}_{N-1|k+1}\bigr] \\
&= \bigl[ \bm{u}^{*}_{1|k},\bm{u}^{*}_{2|k},\cdots, \bm{u}^{*}_{N-1|k},\bm{F}{\bm{z}}^{*}_{N|k}\bigr] 
\end{aligned}
\end{equation}
where $\bm{u}_{N-1|k+1} = \bm{F}{\bm{z}}^{*}_{N|k}$ is the terminal control law. Furthermore, let $\mu_k=0$, which yields ${\bm{z}}_{i|k+1}={\bm{z}}^*_{i+1|k},i=0,1,\dots,N-1$. It holds on to that:
\begin{equation}
\begin{aligned}
V_F({\bm{z}}_{N|k+1} ) &=  V_F\bigl(\bm{A}{\bm{z}}_{N-1|k+1} + \bm{B}{\bm{u}}_{N-1|k+1}\bigr) \\
&= V_F\bigl(\bm{A}_F{\bm{z}}^*_{N|k} \bigr) \\
&\leq  V_F\bigl({\bm{z}}^*_{N|k}\bigr) \\
& \leq \sigma
\end{aligned}
\end{equation}
Thus ${\bm{z}}_{N|k+1} \in \mathbb{Z}_F$, which completes the proof.
\end{IEEEproof}

  \begin{theorem} \label{theorem_ISS_z} 
    Under Assumption 1 and for trajectories evolving in the compact set $\mathbb{X}$, the optimal predicted initial state ${\bm{z}}^{*}_{0|k}$ is ISS, i.e., there exist a $\mathcal{KL}$ function  $\alpha_z(\cdot)$ and a $\mathcal{K}$ function $\beta_z(\cdot)$, such that:
    \begin{equation} \label{ISS_z}
    \Vert {\bm{z}}^{*}_{0|k} \Vert \leq \alpha_z(\Vert {\bm{z}}_{0|0}\Vert,k) + \beta_z(\Vert \bm{e}_{[0,k]} \Vert), k\in\mathbb{N}_{\geq1}
    \end{equation}
    where $\Vert \bm{e}_{[0,k]} \Vert = \sup_{0\le i\le k}\Vert \bm{e}_i \Vert$.
  \end{theorem} 

\begin{IEEEproof} 
Consider the sub-optimal control input sequence $\bm{U}_{k+1} = \bigl[ \bm{u}^{*}_{1|k},\bm{u}^{*}_{2|k},\cdots, \bm{u}^{*}_{N-1|k},\bm{F}\bm{x}^{*}_{N|k}\bigr]$, and sub-optimal predicted initial state ${\bm{z}}_{0|k+1} = {\bm{z}}^*_{1|k}$ at time $k+1$. 

Define the candidate Lyapunov function at time $k$:

\begin{equation}
V\bigl(  {\bm{z}}^*_{0|k}\bigr) = J\bigl (  {\bm{z}}^*_{0|k}, \bm{U}^*_{k}\bigr)
\end{equation}
with
\begin{equation}
\alpha_1( \Vert {\bm{z}}^*_{0|k} \Vert ) \leq V\bigl(  {\bm{z}}^*_{0|k}\bigr) \leq \alpha_2( \Vert {\bm{z}}^*_{0|k}\Vert  )
\end{equation}
where  $\alpha_1(\cdot),\alpha_2(\cdot)$ represent $\mathcal{K}_{\infty}$-functions, respectively. 
It holds that:
\begin{equation}
\begin{aligned}
V\bigl(  {\bm{z}}&^*_{0|k+1}\bigr) - V\bigl( {\bm{z}}^*_{0|k}\bigr) \\
&= J\bigl (  {\bm{z}}^*_{0|k+1}, \bm{U}^*_{k+1}\bigr) - J\bigl (  {\bm{z}}^*_{0|k}, \bm{U}^*_{k}\bigr) \\
&\leq J\bigl (  {\bm{z}}_{0|k+1}, \bm{U}_{k+1}\bigr) - J\bigl (  {\bm{z}}^*_{0|k}, \bm{U}^*_{k}\bigr) \\
&= l \bigl({\bm{z}}_{N|k+1}, {\bm{u}}_{N|k+1}\bigr)- l \bigl({\bm{z}}^*_{0|k}, {\bm{u}}^*_{0|k}\bigr) \\
&\quad + V_F({\bm{z}}_{0|k+1}) - V_F({\bm{z}}^*_{0|k}) + l_r \bigl({\bm{z}}_{0|k+1}\bigr)\\
&\quad  -l_r( {\bm{z}}^*_{0|k} ) \\
&\leq {\bm{z}}^{*\top}_{N|k} \bigl(  \bm{A}_F^{\top} \bm{P} \bm{A}_F - \bm{P} + \bm{Q} +\bm{F}^{\top} \bm{R} \bm{F} \bigr){\bm{z}}^*_{N|k} \\
&\quad- \Vert {\bm{z}}^*_{0|k} \Vert_Q^2 + \gamma_e\Vert \bm{e}_{k+1} \Vert_2^2\\
&\leq - \Vert {\bm{z}}^*_{0|k} \Vert_Q^2 - \kappa\Vert {\bm{z}}^*_{N|k} \Vert_P^2 + \gamma_e\Vert \bm{e}_{k+1} \Vert_2^2
\end{aligned}
\end{equation}
Hence,
\begin{equation} \label{V1}
    V\bigl(  {\bm{z}}^*_{0|k+1}\bigr) - V\bigl(  {\bm{z}}^*_{0|k}\bigr) \leq -\alpha_3\bigl(\Vert {\bm{z}}^*_{0|k}\Vert\bigr) + \beta\bigl( \Vert\bm{e}_{k+1}\Vert\bigr)
\end{equation}
where $\alpha_3(\cdot)$ represent $\mathcal{K}_{\infty}$-functions respectively, $\beta(\cdot)$ represents $\mathcal{K}$-function. Thus $V\bigl(  {\bm{z}}^*_{0|k}\bigr)$ is an ISS-Lyapunov function with respect to ${\bm{z}}^*_{0|k}$. According to \cite{jiang2001input}, the optimal predicted initial satisfies \eqref{ISS_z}, which completes the proof.

\end{IEEEproof}

While Theorem \ref{theorem_ISS_z} shows that the predicted state ${\bm{z}}^*_{0|k}$ satisfies an ISS property under bounded uncertainties, this result alone does not fully capture the behavior of the actual lifted state, i.e. ${\bm{\Phi}}({\bm{x}}_k)$. In practice, the actual lifted state is influenced by additional uncertainties that are absent from the prediction. In other words, although the ISS property of the predicted state provides an important baseline, it is necessary to extend the analysis to the actual lifted state and verify that it also remains stable under the same bounds. This ensures that all sources of error are effectively managed throughout the system dynamics.

  \begin{theorem} \label{theorem_ISS_phi}

    Under the same compact-set assumptions, the lifted Koopman state is ISS, i.e. 
    \begin{equation} \label{ISS_phi}
    \Vert \bm{\Phi}({\bm{x}}_k) \Vert \leq \alpha_\phi(\Vert \bm{\Phi}({\bm{x}}_0)\Vert,k) + \beta_\phi(\Vert \bm{e}_{[0,k]} \Vert), k\in\mathbb{N}_{\geq1}
    \end{equation}
  \end{theorem} 

\begin{IEEEproof}
According to \eqref{MPC_cons_ISS}, the true lifted state $\bm{\Phi}(\bm{x}_k)$ and the optimal predicted initial state ${\bm{z}}^*_{0|k}$ satisfy
\begin{equation} \label{eta}
\bigl\Vert \bm{\Phi}( \bm{x}_k )-{\bm{z}}^*_{0|k} \bigr\Vert = \bigl\Vert \mu^*_k\bm{e}_k \bigr\Vert \leq \eta\left( \Vert \bm{e}_{[0,k]} \Vert \right)
\end{equation}
where $\eta(\cdot)$ represent a $\mathcal{K}$-function. 
By Theorem \ref{theorem_Rec}  and Assumption \ref{ASS1}, the closed-loop trajectory remains feasible and bounded. Therefore, since the lifting map $\bm{\Phi}(\cdot)$ is continuous on $\mathbb{X}$, both $\bm{\Phi}(\bm{x}_{k})$ and $\bm{z}^*_{0\mid k}$ remain in a compact subset of the feasible region of the MPC problem.

For the finite-horizon QP in \eqref{QP}, the optimal value function \(V(\cdot)\) is continuous and piecewise quadratic with respect to the feasible initial lifted state. Therefore, \(V(\cdot)\) is locally Lipschitz on this feasible region. Hence, there exists a constant \(L_1>0\) such that
\begin{equation}
\bigl| V\bigl(\bm{\Phi}(\bm{x}_k)\bigr) - V({\bm{z}}^*_{0|k}) \bigr| \leq L_1\eta\left ( \Vert \bm{e}_{[0,k]} \Vert\right )
\end{equation}

Now, consider the difference along the true state trajectory. According to \eqref{V1}, we have
\begin{equation}
\begin{aligned}
&V\bigl( \bm{\Phi} (\bm{x}_{k+1})\bigr)
 - V\bigl(\bm{\Phi} (\bm{x}_{k})\bigr)\\
&\quad\le
 \bigl|V\bigl(\bm{\Phi} (\bm{x}_{k+1})\bigr)
        - V({\bm{z}}^*_{0\mid k+1})\bigr|\\
&\quad\quad{}+
 \bigl(V({\bm{z}}^*_{0\mid k+1}) - V({\bm{z}}^*_{0\mid k})\bigr)+
 \bigl| V({\bm{z}}^*_{0| k})
        - V\bigl(\bm{\Phi} (\bm{x}_{k})\bigr)\bigr|\\
&\quad\leq -\alpha_3\bigl(\Vert {\bm{z}}^*_{0|k}\Vert\bigr) + \beta\bigl( \Vert\bm{e}_{k+1}\Vert\bigr) + L_1\eta\bigl(\Vert \bm{e}_{[0,k+1]} \Vert\bigr)\\
&\quad\quad{} + L_1\eta\bigl(\Vert \bm{e}_{[0,k]} \Vert\bigr)
\end{aligned}
\end{equation}
Since the supremum over $0\leq i \leq k$ is bounded by the supremum over $0\leq i \leq k+1$, it can be obtained that
\begin{equation} \label{V2}
\begin{aligned}
V\bigl( \bm{\Phi} &(\bm{x}_{k+1})\bigr) - V\bigl(\bm{\Phi} (\bm{x}_{k})\bigr) \\
&\leq - \alpha_3\bigl( \Vert{\bm{z}}^*_{0|k} \Vert\bigr) + 2L_1\eta\left ( \Vert \bm{e}_{[0,k+1]} \Vert \right ) + \beta\bigl( \Vert\bm{e}_{k+1}\Vert\bigr)
\end{aligned}
\end{equation}
Using the triangle inequality, it can be obtained from \eqref{eta}:
\begin{equation}
\|\bm{z}^*_{0\mid k}\| \geq \Vert \bm{\Phi} (\bm{x}_{k}) \Vert - \eta\left(\Vert \bm{e}_{[0,k]} \Vert\right)
\end{equation}
For small prediction error $\bm{e}_k$, the $\mathcal{K}_{\infty}$-function $\alpha_3(\cdot)$ can absorb this deviation, that is
\begin{equation}
\begin{aligned}
\alpha_3\bigl( \Vert{\bm{z}}^*_{0|k} \Vert\bigr) &\geq \alpha_3\left( \Vert \bm{\Phi} (\bm{x}_{k}) \Vert - \eta\left(\Vert \bm{e}_{[0,k]} \Vert\right) \right)\\
&\approx \alpha_3\bigl( \Vert \bm{\Phi} (\bm{x}_{k}) \Vert\bigr) - L_2\eta\left(\Vert \bm{e}_{[0,k]} \Vert\right)
\end{aligned}
\end{equation}
with a Lipschitz constant $L_2>0$. Thus, \eqref{V2} can be rewritten as
\begin{equation}
\begin{aligned}
V\bigl(\bm{\Phi}(\bm{x}_{k+1})\bigr)
- V\bigl(&\bm{\Phi}(\bm{x}_{k})\bigr)
\\ &\le
-\alpha_3\bigl(\Vert \bm{\Phi}(\bm{x}_k) \Vert\bigr)+\hat{\beta}\left(\Vert \bm{e}_{[0,k+1]} \Vert\right)
\end{aligned}
\end{equation}
with $\hat{\beta}(\cdot)=\beta(\cdot) + 2L_1\eta(\cdot)+L_2\eta(\cdot)$. Then \eqref{ISS_phi} can be obtained by standard ISS arguments, which completes the proof.
\end{IEEEproof}
This regional ISS property is appropriate for the present problem because the control objective is robust practical stabilization under bounded uncertainty, rather than global asymptotic regulation of an ideal disturbance-free system.

\section{Simulation Results} \label{Sec_Simulation}
This section presents numerical simulations that assess the efficacy of the proposed control strategy for stabilizing the attitude of the combined spacecraft in the presence of dynamic uncertainty.
\subsection{Simulation Platform and Parameters}
In this study, a comprehensive simulation framework was developed using the MATLAB/Simscape Multibody environment to model the attitude dynamics of the combined spacecraft system. For simplicity, the servicer is represented as a rigid cube-shaped satellite equipped with two three-link robotic arms, whereas the target is similarly modeled as a rigid cubesat equipped with a pair of solar panels and a docking ring.

The inertial moments of the servicer and target are respectively set as $\bm{J}_s=\text{diag}(405,405,405)\mathrm{kg}\cdot \mathrm{m^2}$, $\bm{J}_t=\text{diag}(36.8,37.5,36.8)\mathrm{kg}\cdot \mathrm{m^2}$\cite{liu2024attitude}.
The mass of the servicer and target are selected as $m_s = 1080\mathrm{kg}$ and $m_t = 75\mathrm{kg}$.
The inertia of the entire combined spacecraft is computed as
\[
\bm{J}_c=\begin{bmatrix}
    986.1 & 1 & -1\\
     1& 573.2 & -1\\
      -1 &  -1 & 1116
\end{bmatrix}\mathrm{kg}\cdot \mathrm{m^2}.
\]
The initial quaternion of the combined spacecraft is $\begin{bmatrix}0.5 & 0.5 & 0.5 & 0.5\end{bmatrix}^{\top}$, the initial angular velocity is $\begin{bmatrix}0.02 & -0.01 & -0.01\end{bmatrix}^{\top}\mathrm{rad/s}$. The desired quaternion and angular velocity are set as $\begin{bmatrix}1 & 0 & 0 & 0\end{bmatrix}^{\top}$ and $\begin{bmatrix}0 & 0 & 0\end{bmatrix}^{\top}\mathrm{rad/s}$, respectively. The constraint of control torques is set as $\vert u_{\text{max}}\vert\leq20\mathrm{N}\cdot \mathrm{m}$.

The weighting matrix of MPC are chosen as $\bm{Q}=10^2\textbf{I}_7$ and $\bm{R}=10^{-2}\textbf{I}_3$. The terminal weighting matrix $\bm{P}$ is solved by \eqref{opt_P} using MOSEK \cite{aps2019mosek}, which also yields the bound for the terminal set as $\sigma = 208.68$. The regularization coefficient is set as $\gamma_e=1e6$. The configuration time is set as 200$s$ and the simulation step is 0.1$s$.

To improve reproducibility, the main hyper-parameters of the proposed method are summarized in Table \ref{tab:hyperparameters}. In particular, \(N_\phi\) was selected according to the offline accuracy-complexity comparison in Fig. \ref{fig:Com_diff_model_err} and Table \ref{tab:perf_models}. The remaining parameters were chosen based on explicit trade-offs among adaptation speed, noise attenuation, and computational burden.

\begin{table}[htb]
  \centering
  \caption{PRACTICAL HYPER-PARAMETER SELECTION GUIDELINES AND VALUES}
  \label{tab:hyperparameters}
  \begin{tabular}{llll}
    \hline\hline
    \textbf{Parameter} & \textbf{Role} & \textbf{Main trade-off}  & \textbf{Value} \\
    \hline
    $N_{\phi}$ & \makecell[l]{Koopman \\truncation order}  & \makecell[l]{accuracy \\ vs. runtime}   & 5 \\
    
    $N_{w}$ & \makecell[l]{data window length}   & \makecell[l]{noise suppression \\ vs. adaptation speed}  & 10 \\
    
    $N$ & \makecell[l]{prediction horizon}   & \makecell[l]{performance \\ vs. online QP size}  & 15 \\
    
    $\lambda$ & \makecell[l]{forgetting factor}    & \makecell[l]{memory vs.\\ responsiveness}  & 0.965 \\
    
    $\rho$ & \makecell[l]{learning rate}    & \makecell[l]{convergence speed \\ vs. oscillation} & 1e-7 \\
    
    $\gamma$ & \makecell[l]{regularization \\ weight}    & \makecell[l]{robustness vs. bias}  & 1e-3 \\
    
    $\beta_1,\beta_2$ & \makecell[l]{Adam first/second \\ moment}    & \makecell[l]{standard optimizer \\ setting} & 0.9,0.999 \\
    
    \hline\hline
  \end{tabular}
\end{table}

All the simulations were performed on a laptop equipped with an Intel Core i7-12700H CPU, 32 GB of RAM, and an NVIDIA GeForce RTX 3060 GPU.
\subsection{Comparison Study of Nominal Koopman Model}
\begin{figure}[!htb]
	\centering
 \includegraphics[width=\linewidth]{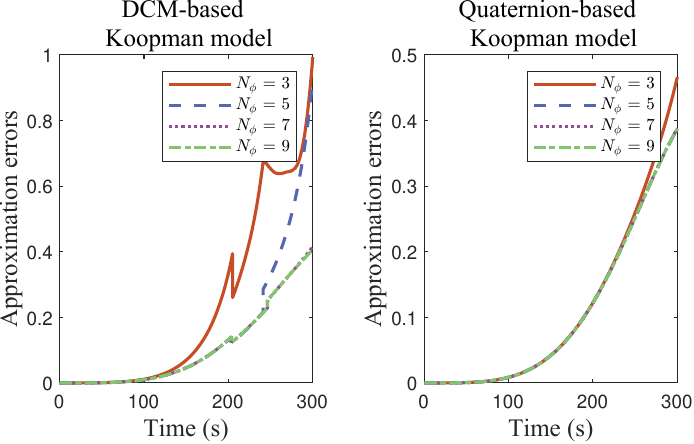}
        \caption{Comparison results of two nominal models. The proposed quaternion-based nominal model demonstrates higher model prediction accuracy.}
        \label{fig:Com_diff_model_err}
\end{figure}
Before validating the effectiveness of the control strategy, it is essential to verify the validity and computational performance superiority of the proposed quaternion-based nominal Koopman model. First, a comparative simulation is conducted against a DCM-based Koopman model. In the simulation, both models are driven by the same time-varying control torque, defined as:
\[
\bm{\tau}(t) = 0.05\begin{bmatrix}7.5 \sin(2t) + 7.5\sin(0.66t) \\ 2.5 \sin(2.2t+\frac{\pi}{2}) + 2.5\sin(0.22t) +1 \\ 5\sin(0.66t) +5\sin(1.52t)-1 \end{bmatrix}\mathrm{N \cdot m}
\]
The approximation error is defined as the 2-norm of the difference between the state predicted by the Koopman model and the true state of the nonlinear dynamics.
The results are illustrated in Fig. \ref{fig:Com_diff_model_err}. As depicted, increasing the truncation index $N_{\phi}$ reduces the approximation error for both modeling approaches. For the same truncation index $N_{\phi}$, the approximation errors of the proposed model are significantly smaller than those of the DCM-based model. This highlights the proposed model's higher fidelity in capturing the spacecraft's nonlinear attitude dynamics.

Next, the computational performance of the proposed model is compared with the DCM-based model and a data-driven model constructed using the EDMD method with Radial Basis Functions (RBFs). The EDMD model is purely data-driven, where the basis functions consist of the original 7 states (quaternion and angular velocity) and 50 RBFs, resulting in a 57-dimensional lifted state space. 

A closed-loop attitude stabilization simulation was performed for all three models. For simplicity, the control strategies of the three models are selected as conventional MPC without adaptation and the uncertainties is not considered, i.e., $\bm{J}_{c0} = \bm{J}_{c}$. 
As shown in Fig. \ref{fig:Com_diff_model}, all three methods can perform the attitude stabilization task. The DCM and quaternion-based models exhibit a faster and smoother response, whereas the RBF-based EDMD model results in a slower, more oscillatory transient, showing the higher model accuracy of the directly calculated model. 
The DCM and quaternion-based models provide a faster and smoother response. This superior performance highlights their higher fidelity compared to the data-driven RBF-EDMD approach, which yields a slower and more oscillatory transient.
It should be noted that the control torque from the DCM-based model is inverted because its error metric, defined using lifted states $\text{vec}(\bm{R})$, possesses dynamics that diverge from the true physical attitude error. As a result, the MPC generates a control torque with a direction opposite to the physically intuitive one required for stabilization.
\begin{figure}[!htb]
	\centering
	\includegraphics[width=\linewidth]{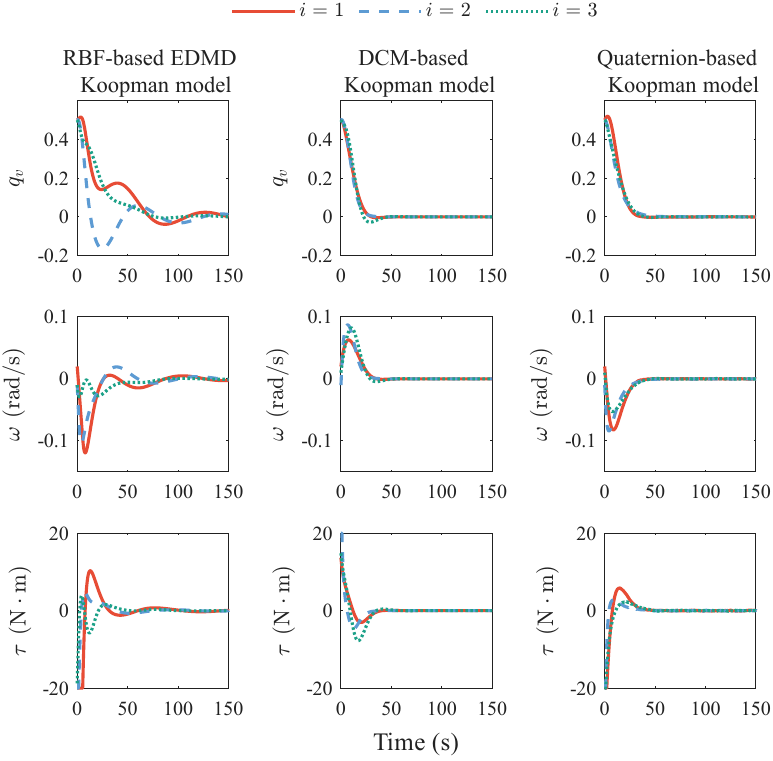}
        \caption{Comparison results of the considered nominal models. Taking initial quaternion $\left[  0.5 \ 0.5 \ 0.5 \ 0.5 \right]^{\top}$ and initial angular velocity $\left[ 0.02 \ -0.01 \ -0.01 \right]^{\top}\mathrm{rad/s}$ for example, all three methods successfully stabilize the attitude, while the analytically derived models (DCM and quaternion-based models) exhibit superior model accuracy.}
        \label{fig:Com_diff_model}
\end{figure}
To further evaluate and compare the robustness of the different models, a set of 50 Monte Carlo simulations were conducted.
In each run, the system was initialized with a random attitude and a random angular velocity (in the range of ±0.02 rad/s)
Figure \ref{fig:boxplot} presents the resulting box plots for two key metrics: settling time (the time required to enter and remain within a $5\%$ error margin) and total control effort $( \int\Vert  \bm{\tau} \Vert_2^2\mathrm{d}t )$. It shows that the analytical models (DCM-based and Quaternion-based) vastly outperform the data-driven RBF-EDMD model. They are not only significantly faster and more efficient, but their tightly clustered results also demonstrate their highly consistent performance and strong robustness.
\begin{figure}[!htb]
	\centering
	\includegraphics[width=\linewidth]{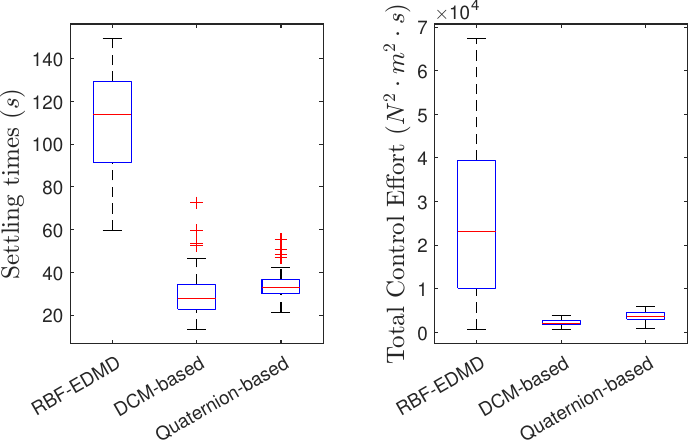}
        \caption{Box plots for settling time and total control effort of the considered nominal models. Both the DCM-based and Quaternion-based methods achieve significantly faster settling times and lower control effort compared to the RBF-EDMD model. The tight distributions for the analytical models also demonstrate their greater robustness and predictability.}
        \label{fig:boxplot}
\end{figure}

\begin{table*}[ht]
  \centering
  \caption{Computational Performance Metrics of the Two Models}
  \label{tab:perf_models}
  \begin{tabular}{llllll}
    \hline\hline
    \textbf{Model} & \textbf{\makecell[l]{Model truncation \\ index}} & \textbf{\makecell[l]{Number of \\ lifted states}} & \textbf{\makecell[l]{Simulation \\ runtime (s)}} & \textbf{\makecell[l]{Average step \\ runtime (s)}}  & \textbf{Memory usage} (KB) \\
    \hline
    \multirow{2}{*}{Quaternion-based model} 
    & 5 & 27   & \textbf{16.85}   & \textbf{0.011} & \textbf{13368}   \\
    & 7   & 35   & \textbf{18.30}   & \textbf{0.012} & \textbf{13415}   \\
    \hline
    \multirow{2}{*}{DCM-based model \cite{Chen2023}} 
    & 5 & 45   & 24.81   & 0.017 & 14731   \\
    & 7 & 63   & 36.95   & 0.024 & 15620   \\
    \hline
    RBF-based EDMD model \cite{korda2018linear} & / & 57   & 27.25   & 0.018 & 17418   \\
    \hline\hline
  \end{tabular}
\end{table*}

Table \ref{tab:perf_models} records the computational performance metrics of the three models in terms of simulation runtime, average step runtime, and memory usage. Note that the "Simulation runtime" represents the wall-clock time cost of the simulation, and "Memory usage" represents the peak RAM, which can be obtained via MATLAB Profiler. 
The results indicate that the proposed quaternion-based model is significantly more efficient than the other two models, both in terms of computation time and memory usage. 
Furthermore, when the model fidelity is increased by raising $N_{\phi}$, the simulation runtime for the proposed model increases more slowly than that of the DCM-based model.
This difference arises because, for the same value of $N_{\phi}$, the DCM-based model requires more lifted states to represent the original nonlinear dynamics, leading to greater computational overhead.

\subsection{Adaptive Attitude Stabilization} \label{Section-V-C}
In this subsection, we first evaluate the proposed PAKMPC under a representative post-capture attitude-stabilization scenario and compare it with non-adaptive MPC, KMAE-AKMPC \cite{chen2024learning}, and Incremental Nonlinear MPC (INMPC) \cite{chen2025incremental}. This subsection focuses on the transient and steady-state behavior of the four controllers under a fixed operating condition.
The KAME-AKMPC strategy utilizes a linear MPC based on a nominal model, which is initially established using the EDMD method and then continuously updated online via a Recursive Least Squares (RLS) algorithm. The INMPC strategy employs a nonlinear MPC framework that operates directly on an incremental, data-driven model constructed from system data.

For the nominal Koopman model, the truncation index is selected as $N_{\phi}=5$, then a $4N_{\phi}+7=27$-dimensional linear model is obtained. The nominal inertia of combined systems is set as that of the servicer, i.e., $\bm{J}_{c0} = \bm{J}_{s}=\mathrm{diag}(405,405,405)\mathrm{kg}\cdot \mathrm{m^2}$. The active attitude control law for the target is formulated in the form of a PD controller:
\begin{equation}
\bm{u}_{t}=-\bm{K}_{pt}\bm{q}_{vt} - \bm{K}_{dt} \bm{\omega}_{t}
\end{equation}
where $\bm{K}_{pt} = 0.02\bm{J}_t$, $\bm{K}_{dt} = 0.05\bm{J}_t$, $\bm{q}_{vt}$ denotes the vector part of initial quaternion of the target and $\bm{\omega}_{t}$ denotes the angular velocity of the target.
To reflect implementation uncertainty, additive external disturbance torques and measurement noise are included in the comparative simulations. Specifically, the disturbance torque is modeled as $\bm{\tau}_d(t)=\begin{bmatrix}
0.1\sin(0.075 t)\ 
-0.2\sin(0.1 t)\ 
0.15\sin(0.1 t)\end{bmatrix}^{\top}\mathrm{N} \cdot \mathrm{m}$, and the state measurements $\bm{q}_v$ and $\bm{\omega}$ are corrupted by additive zero-mean Gaussian white noise following $\mathcal{N}(\mathbf{0},10^{-4}\mathbf{I}_3)$.

The adaptive-learning hyper-parameters used in this subsection follow the practical selection guidelines summarized in Table \ref{tab:hyperparameters}.
\begin{figure*}[!htb]
    \centering 
    \subfloat[]{
        \centering
		\includegraphics[width=0.2\textwidth]{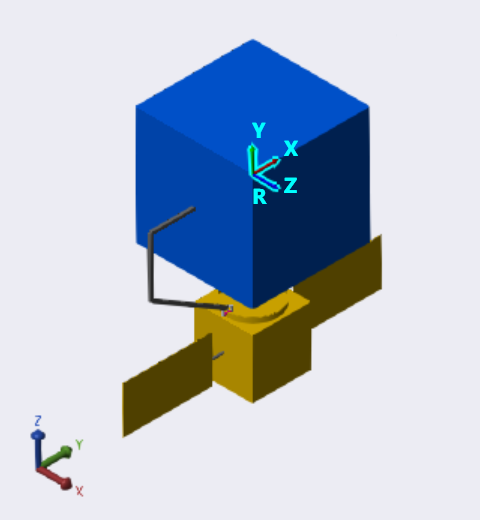}}
    \hfill
    \subfloat[]{
			\centering
			\includegraphics[width=0.2\textwidth]{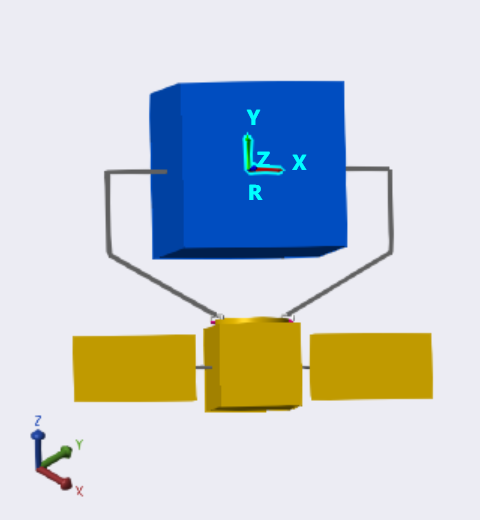}}
    \hfill
    \subfloat[]{
			\centering
			\includegraphics[width=0.2\textwidth]{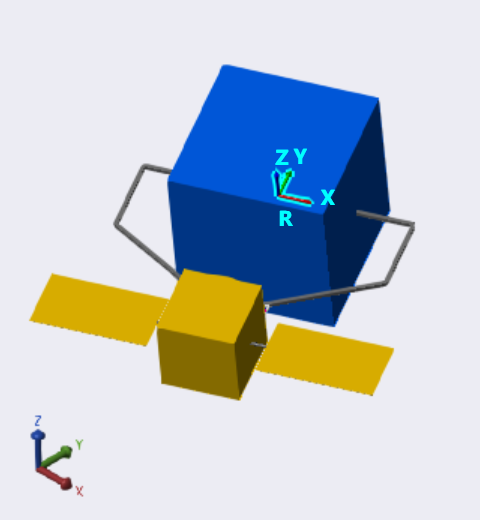}}
    \hfill
    \subfloat[]{
			\centering
			\includegraphics[width=0.2\textwidth]{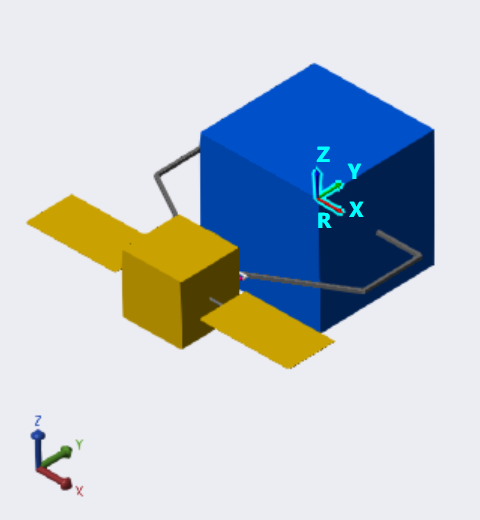}}
    \caption{Simulation snapshots of the configuration of the combined spacecraft under the proposed control strategy. The subfigures show the mission timeline: (a) $t = 0$s (initial states). (b) $t = 10$s. (c) $t = 20$s. (d) $t = 50$s (substantial stabilization). It can be observed that the attitude stabilization task has been effectively achieved at $t = 50$s.}
    \label{fig:satellite_snapshots}
\end{figure*}
\begin{figure}[!htb]
	\centering
	\includegraphics[width=\linewidth]{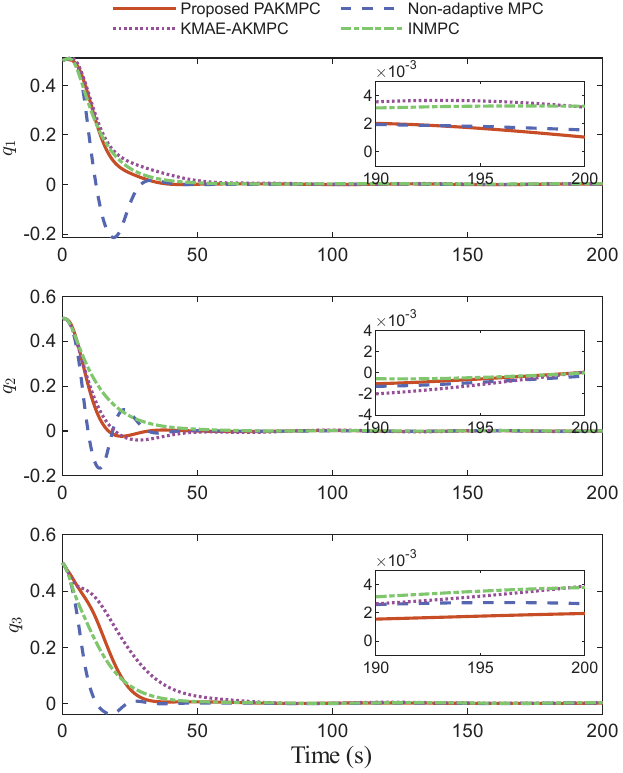}
        \caption{The time responses of $\bm{q}_v$ under different control strategies. The proposed PAKMPC achieves a convergence speed comparable to INMPC and faster than KMAE-AKMPC, while avoiding the significant overshoot seen with the non-adaptive MPC. The magnified subplots of the steady-state response highlight that the proposed PAKMPC also achieves the smallest steady-state error.}
        \label{fig:q}
\end{figure}
\begin{figure}[!htb]
	\centering
	\includegraphics[width=\linewidth]{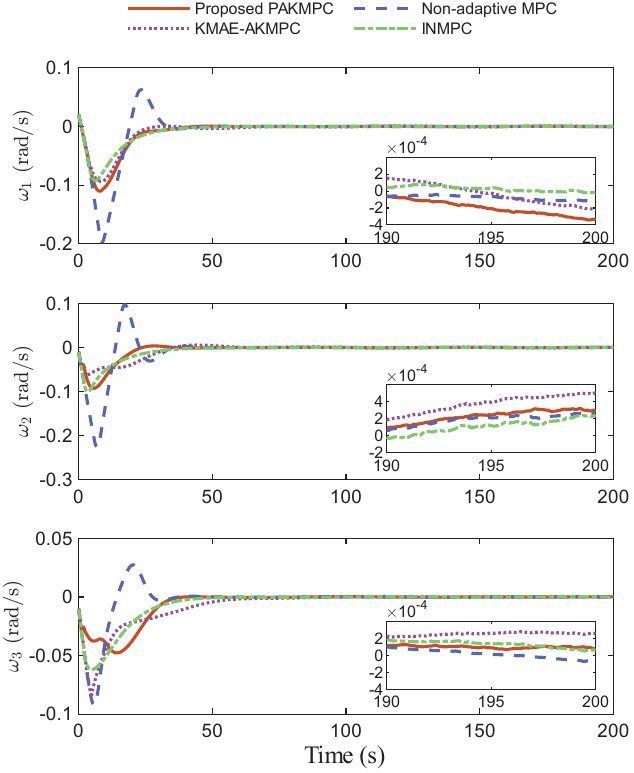}
         \caption{The time responses of $\bm{\omega}$ under different control strategies. Similar to Fig. \ref{fig:q}, the proposed PAKMPC and the INMPC perform better.}
        \label{fig:omega}
\end{figure}

\begin{figure}[!htb]
	\centering
	\includegraphics[width=\linewidth]{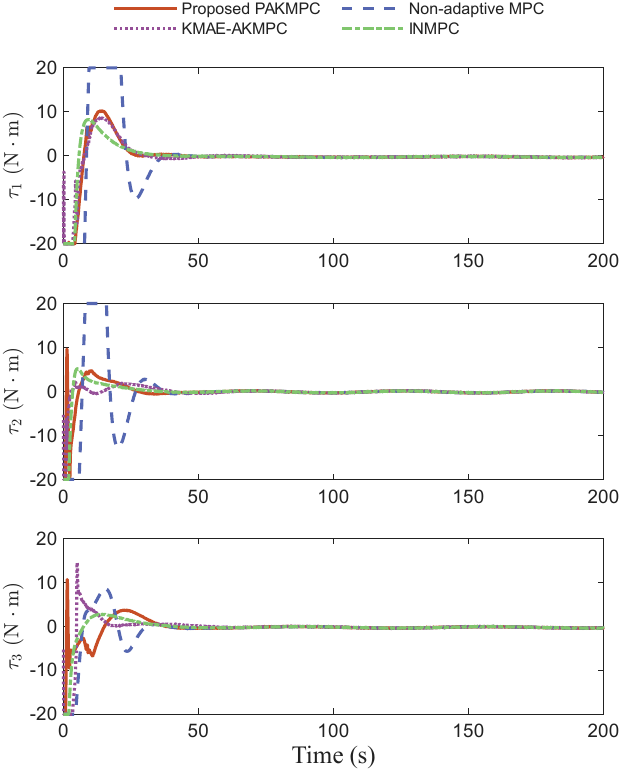}
        \caption{The time responses of control torques $\bm{\tau}$. 
        The non-adaptive MPC suffers from severe control saturation due to the uncompensated uncertainty.}
        \label{fig:tau}
\end{figure}

\begin{figure}[!htb]
	\centering
	\includegraphics[width=\linewidth, trim={0cm 1cm 0cm 0cm}, clip]{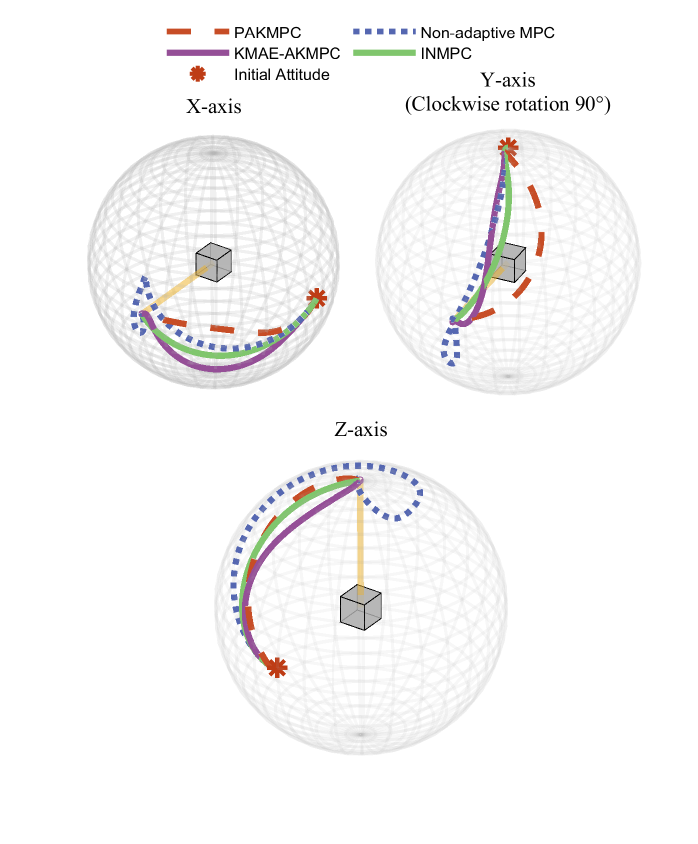}
        \caption{3D attitude trajectories of the combined spacecraft.}
        \label{fig:3d_traj}
\end{figure}

The snapshots of Fig. \ref{fig:satellite_snapshots} provide an overview of the configuration of the combined spacecraft under the proposed PAKMPC control strategy, illustrating the successful execution of the entire maneuver from the initial states to final stabilization.
The time responses of attitude states and control torques are plotted in Figs. \ref{fig:q} - \ref{fig:tau}.
As shown in Fig. \ref{fig:q} and \ref{fig:omega}, all the control strategies are capable of achieving attitude stabilization of the combined system under dynamic uncertainty. Notably, the proposed PAKMPC achieves faster attitude stabilization than KMAE-AKMPC, while the non-adaptive MPC appears to converge most rapidly but exhibits severe overshoot. This behavior is more clearly visualized in the 3D attitude trajectories plotted in Fig. \ref{fig:3d_traj}, where the significant overshoot from the non-adaptive MPC's path is clearly observable. As Fig. \ref{fig:tau} reveals, the faster coverage speed is produced by a larger and saturated control torque. The proposed PAKMPC achieves the lowest state error among all methods, with accuracy comparable to the Non-adaptive MPC. 

\begin{remark}
The apparently fast convergence of the non-adaptive MPC is achieved at the expense of much larger actuation, as also seen from the sustained saturation tendency in Fig. \ref{fig:tau}. In other words, its seemingly competitive steady-state accuracy is obtained by excessive control effort rather than by a better prediction model.
\end{remark}

\begin{table*}[htb]
  \centering
  \caption{Comprehensive performance evaluation of different strategies}
  \label{tab:perf_methods}
  \begin{tabular}{llllll}
    \hline\hline
    \textbf{Strategy} & \makecell{\textbf{Simulation} \\ \textbf{runtime} (s)}& \makecell{\textbf{Average step} \\ \textbf{runtime} (s)}& \makecell{\textbf{Steady-state} \\ \textbf{error} $\Vert\bm{q}_{v}\Vert$} & \makecell{\textbf{Steady-state} \\ \textbf{error} $\Vert\bm{\omega}\Vert$} & \makecell[l]{\textbf{Control effort} \\ $( \int\Vert  \bm{\tau} \Vert_2^2\mathrm{d}t )$}\\
    \hline
    PAKMPC & \textbf{20.11}   & \textbf{0.010} & $\bm{2.1e^{-3}}$  & {$2.2e^{-4}$} & {$4.8e^{4}$}\\
    Non-adaptive MPC & 18.79   & 0.009 & {$2.8e^{-3}$} & {$2.1e^{-4}$} & {$1.8e^{5}$}\\
    KMAE-AKMPC & 33.87   & 0.017 & {$3.4e^{-3}$} & {$5.3e^{-4}$} & {$4.8e^{4}$}\\
    INMPC & 70.49   & 0.035 & {$3.9e^{-3}$} & {$2.3e^{-4}$} & {$4.7e^{4}$}\\
    \hline\hline
  \end{tabular}
\end{table*}

Table \ref{tab:perf_methods} summarizes the computational and control performance of the considered strategies under sensor noise and external disturbances. The proposed PAKMPC maintains an average step runtime of 0.010 s, which remains well below the sampling period of 0.1 s and is significantly lower than those of KMAE-AKMPC (0.017 s) and INMPC (0.035 s). In terms of control accuracy, PAKMPC achieves the smallest steady-state attitude error among all compared strategies, with $\Vert\bm{q}_{v}\Vert=2.1e^{-3}$. Its steady-state angular-velocity error $\Vert\bm{\omega}\Vert=2.2e^{-4}$ remains comparable to the best-performing baselines under the same disturbance and noise conditions. In addition, the required control effort is substantially lower than that of the non-adaptive MPC, indicating that the proposed method preserves good robustness and efficiency without relying on excessively large control inputs.
\begin{figure}[t]
	\centering
	\includegraphics[width=\linewidth]{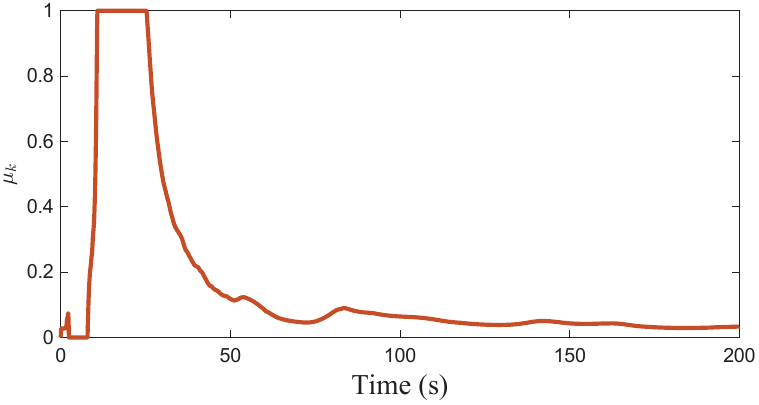}
         \caption{The time response of $\mu_k$.}
          \label{fig:mu}
\end{figure}

\begin{figure}[t]
	\centering
	\includegraphics[width=\linewidth]{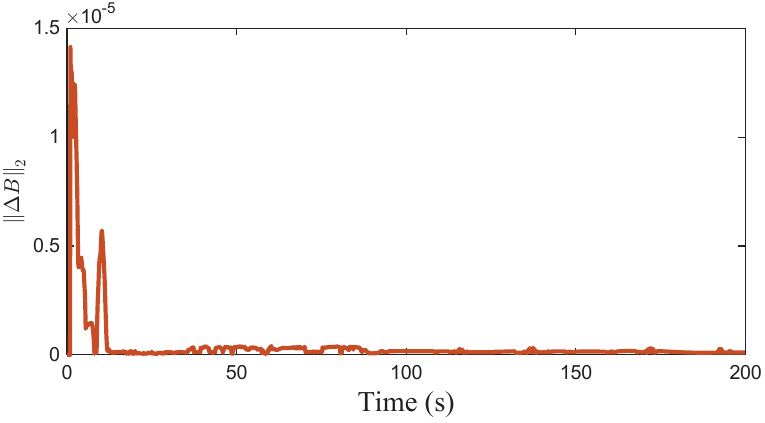}
         \caption{The norm of $\Delta\bm{B}$. }
          \label{fig:del_B_norm}
\end{figure}
Fig. \ref{fig:mu} shows the time responses of the interpolation variable $\mu_k$. The dynamic behavior of $\mu_k$ demonstrates an adaptive weighting between prediction and measurement. Initially, large state errors push $\mu_k$ to its limit of 1. As the system converges, $\mu_k$ smoothly decreases, reflecting an increasing confidence in the model's predictive accuracy.
Fig. \ref{fig:del_B_norm} shows the norm of the uncertainty matrix $\Delta\bm{B}$ at each time step. With the formation of the first complete data window, the norm of $\Delta\bm{B}$ rises rapidly and oscillates as the optimizer seeks the optimal solution, before finally converging.

\subsection{Structured Robustness Evaluation} \label{Section-V-D}
To further assess the robustness and generality of the controller, we conducted a structured parameter sweep while keeping the initial condition, reference command, prediction horizon, and controller hyperparameters fixed. In particular, the target inertia and mass were jointly scaled by factors in {50\%, 75\%, 100\%, 125\%, 150\%}, where 100\% denotes the nominal post-capture target model. Meanwhile, the disturbance torque profile was scaled by factors in {0\%, 50\%, 100\%, 150\%, 200\%}, where 0\% corresponds to the disturbance-free case and 100\% corresponds to the nominal disturbance amplitude.

Fig. \ref{fig:heatmap} compares all four controllers over the same operating range in terms of steady-state quaternion error and total control effort, while Fig. \ref{fig:scatter} summarizes the runtime-performance trade-off of all four controllers over the same set of operating points. The results show that PAKMPC is consistently more robust than the other three baseline metohds and achieves the most favorable overall balance between computational cost and control performance.

\begin{figure}[t]
    \centering
    \subfloat[]{
	   \centering
	       \includegraphics[width=\linewidth]{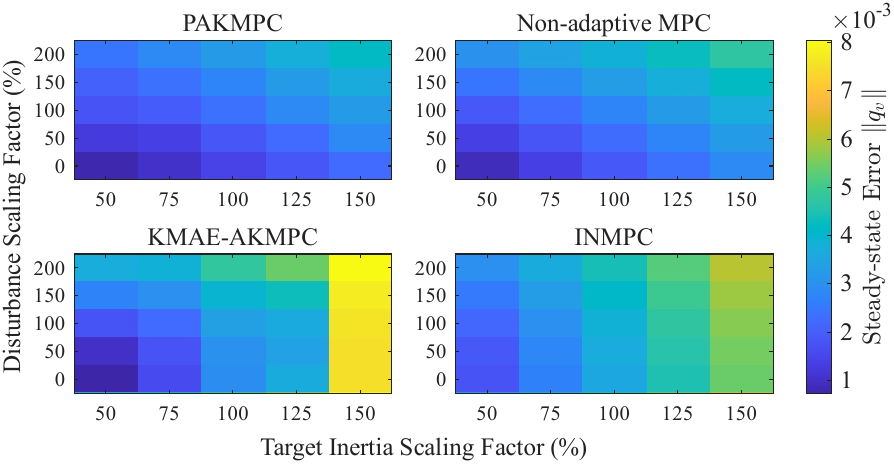}}\\
    \subfloat[]{
	   \centering
	       \includegraphics[width=\linewidth]{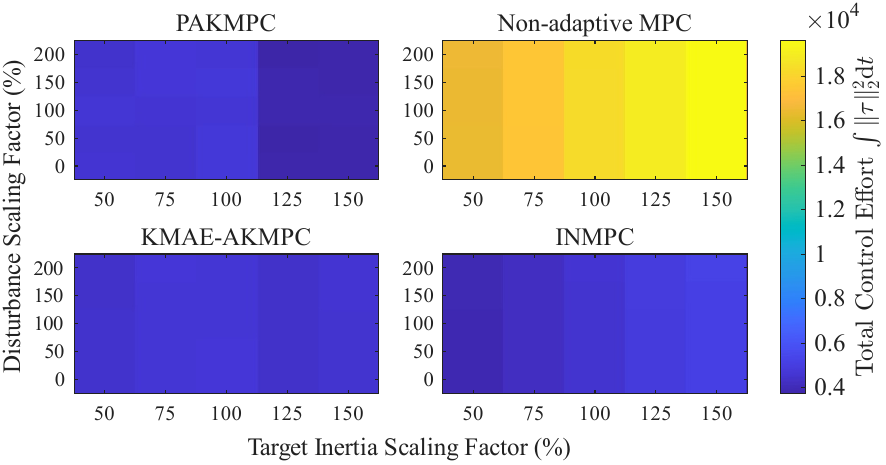}}
    \caption{Structured robustness heatmaps over the uncertainty-disturbance sweep. (a) Steady-state quaternion-vector error $\bm{q}_v$ for PAKMPC, non-adaptive MPC, KMAE-AKMPC, and INMPC. (b) Total control effort over the same operating grid. Across the scanned operating range, PAKMPC shows the most favorable combination of low steady-state attitude error and moderate control effort.}
    \label{fig:heatmap}
\end{figure}

\begin{figure}[t]
\centering
\includegraphics[width=\linewidth]{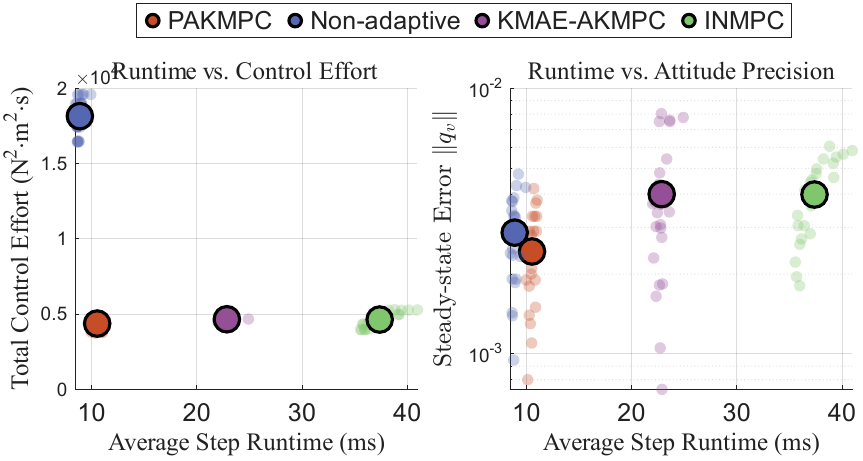}
  \caption{Runtime-performance scatter plots over the structured parameter sweep. The left panel shows runtime versus control effort, and the right panel shows runtime versus steady-state quaternion error. PAKMPC is located closest to the lower-left region, indicating the best overall trade-off among the compared controllers.}
  \label{fig:scatter}
\end{figure}

Although the present results are obtained for a combined-spacecraft scenario, the proposed framework can also be extended to more experimentally accessible platforms. For UAV attitude control, the quaternion-based rigid-body formulation makes the proposed analytical lifting readily adaptable, while the online adaptation can compensate for aerodynamic disturbances and inertia variations \cite{seshasayanan2024robust}. For robotic manipulators, the quaternion-based lifting is not directly reused, but the same physics-based lifting and adaptive Koopman-MPC architecture can be reformulated to address payload-dependent inertia uncertainties and nonlinear couplings \cite{tong2023adaptive, muralidharan2025ground}.

\section{Conclusion} \label{Sec_Conclusion}
In this paper, we proposed a Koopman operator theory-based online adaptive MPC strategy for attitude stabilization of the combined spacecraft with inertial uncertainty. 
The strategy is founded on a physics-based nominal Koopman model, which achieves a structurally concise $(4N_{\phi}+7)$-dimensional linear representation of the nonlinear attitude dynamics, and avoids the higher dimensionality of DCM-based models ($9(N_{\phi}+1)$ dimensions) and EDMD models (up to hundreds of dimensions).
This physics-based model is then augmented by a data-driven recursive online update law, forming a Koopman modeling framework that is both adaptive and physically interpretable. The modeling framework identifies inertial uncertainty and provides precise predictions for the LMPC scheme.
Furthermore, a terminal cost and an invariant set were designed to formally guarantee input-to-state stability and recursive feasibility. The superiority of the overall control framework is demonstrated in comparative high-fidelity simulations, showcasing high computational efficiency and rapid convergence.

It should also be noted that the proposed framework has been validated in a high-fidelity simulation environment, and its extension to other platforms should be understood in a system-dependent sense. The compact Koopman structure developed in this paper is derived from rigid-body quaternion attitude dynamics. Therefore, when applying the same methodology to real platforms such as UAVs or robotic manipulators, additional unmodeled effects, including actuator bandwidth limits, communication delays, flexible modes, and contact forces, may require an augmented lifted model rather than the direct reuse of the present compact structure.
Correspondingly, the terminal set should be constructed in the augmented lifted state space, or replaced by a robust positively invariant/tube-based terminal set when residual uncertainty is treated as a bounded disturbance. In this sense, the present study demonstrates the effectiveness of the proposed physics-based adaptive Koopman MPC for uncertain combined-spacecraft attitude stabilization. Its extension to other platforms with more complex unmodeled dynamics will be investigated through hardware-in-the-loop and experimental validations in future work.

\bibliographystyle{IEEEtran}
\bibliography{bibliography} 

\end{document}